\documentclass[letterpaper]{article} 
\usepackage[margin=1in]{geometry}
\usepackage[hyphens]{url} 
\usepackage{graphicx}
\usepackage[numbers,sort&compress]{natbib}
\usepackage{caption}
\title{Deviate or Not: Learning Coalition Structures \\ with Multiple-bit Observations in Games}
\author{
    Yixuan Even Xu \thanks{Carnegie Mellon University. Email: \texttt{yixuanx@cs.cmu.edu}}
    \and  
    Zhe Feng \thanks{Google Research. Email: \texttt{zhef@google.com}}
    \and
    Fei Fang \thanks{Carnegie Mellon University. Email: \texttt{feif@cs.cmu.edu}}
}
\date{}

\usepackage[
    colorlinks=true,    
    linkcolor=blue,     
    citecolor=blue,     
    urlcolor=blue       
]{hyperref}

\usepackage[linesnumbered,ruled,vlined]{algorithm2e}

\usepackage{amsthm}
\usepackage{amssymb}
\usepackage{amsmath}
\usepackage{mathtools}

\usepackage{enumerate}

\usepackage{subcaption}

\usepackage[capitalise]{cleveref}

\usepackage{threeparttable}
\usepackage{hhline}
\usepackage{multirow}

\usepackage{tikz}
\usetikzlibrary{calc}

\usepackage{xcolor}

\usepackage{thmtools}
\usepackage{thm-restate}
\newtheorem{theorem}{Theorem}[section]
\newtheorem{definition}{Definition}[section]

\newtheorem{lemma}{Lemma}[section]

\newtheorem{corollary}[lemma]{Corollary}

\newcommand{\G}{\mathcal G}
\renewcommand{\S}{\mathcal S}
\renewcommand{\O}{\mathcal O}
\newcommand{\BR}{\mathrm{BR}}
\renewcommand{\P}{\mathcal P}
\newcommand{\B}{\mathcal B}
\newcommand{\A}{\mathcal A}
\renewcommand{\v}{\mathbf v}
\renewcommand{\r}{\mathbf r}

\newcommand{\ind}[1]{\mathbb{I}\hspace{-0.1em}\left[\vphantom{\sum}#1\right]}

\begin{document}

\maketitle

\begin{abstract}
	We consider the Coalition Structure Learning (CSL) problem in multi-agent systems, motivated by the existence of coalitions in many real-world systems, e.g., trading platforms and auction systems. In this problem, there is a hidden coalition structure within a set of $n$ agents, which affects the behavior of the agents in games. Our goal is to actively design a sequence of games for the agents to play, such that observations in these games can be used to learn the hidden coalition structure. In particular, we consider the setting where in each round, we design and present a game together with a strategy profile to the agents, and receive a multiple-bit observation -- for each agent, we observe whether or not they would like to deviate from the specified strategy. 
We show that we can learn the coalition structure in $O(\log n)$ rounds if we are allowed to design any normal-form game, matching the information-theoretical lower bound. For practicality, we extend the result to settings where we can only choose games of a specific format, and design algorithms to learn the coalition structure in these settings. For most settings, our complexity matches the theoretical lower bound up to a constant factor.
\end{abstract}

\section{Introduction}
\label{sec:introduction}

Coalitions and collusive behavior are common in the real world. In auction systems, bidders may form coalitions to coordinate bids and exploit mechanisms for better odds of winning~\cite{milgrom2004putting}. On ridesharing platforms like Uber and Lyft, some drivers might disconnect simultaneously to create an artificial shortage, triggering a price surge they can later benefit from~\cite{hamilton2019uber,Sweeney2019manipulate,Dowling2023drive}. While illegal collusion, like price fixing, is regulated by agencies such as the SEC in the United States, it remains challenging to eliminate, as seen in scandals like the LIBOR Scandal~\cite{Libor}. As AI systems become more advanced and widely used in domains like stock trading and autonomous driving, regulating these AI agents presents a growing challenge as seen by regulators and researchers~\cite{ethicsofAI,qu2023adversarial}.

In response to such concerns, the problem of coalition structure learning (CSL)~\cite{xu2023learning} has recently gained attention. In CSL, a platform aims to figure out which agents are in a coalition by \textit{designing a small number of games} for the agents to play and \textit{make observations on the outcomes} of these games. Notably, more than one coalition may exist simultaneously. In practice, the platform could represent a regulator, like an auction or ridesharing platform, aiming to detect collusion or design better mechanisms.

However, the existing study on CSL relies on a strong assumption that the platform can only observe whether a specified strategy profile is a Nash equilibrium (NE) in each round. With this single-bit observation oracle and other assumptions, optimal algorithms require $\Theta(n \log n)$ rounds of games to learn the coalition structure, where $n$ is the number of agents.
Such an observation oracle is too restrictive for two reasons. First, for real-world systems with a large number of agents, $\Theta(n \log n)$ rounds is still prohibitively large. Second, it may be possible for the platform to know which agents would deviate from the specified strategy profile. For example, a ridesharing platform could observe which drivers disconnect from the platform given a pricing scheme.

In this paper, we aim to address this limitation and consider a \textit{multiple-bit observation oracle}: the platform can observe an $n$-dimensional binary vector in each round, where the $i$-th bit indicates whether or not agent $i$ would like to deviate from the specified strategy profile. 
This oracle allows the platform to gain $n$ bits of information in each round, which is significantly more informative than the single-bit observation oracle, allowing the platform to learn the coalition structure with much fewer rounds. However, from a technical perspective, switching from the single-bit observation oracle to the multiple-bit observation oracle fundamentally changes the problem. With the single-bit observation, the platform can choose the game to use to get the next bit of observation, which means it is easy to ensure that each bit of observation brings in information. 
However, with multi-bit observation, the platform will get $n$-bit information in a batch, and it is often inevitable that these $n$-bits contain repetitive or highly correlated information. Therefore, the key challenge in algorithm design is how to design the games so that the $n$ bits of observation are more de-correlated and contain more information.

\subsection{Our Results}
\label{subsec:our_results}

\begin{table*}[t]
	\centering
	\begin{tabular}{|c|c|c|c|c|}
		\hline
		Type of Games & Lower Bound & Complexity & Section \\
		\hline
        Normal-Form & $\log_2 n-O(\log\log n)$ & $\log_2 n + 2$ & \cref{sec:normal_form_games} \\
        \hline
        Congestion & $\log_2 n-O(\log\log n)$ & $\log_2 n + 2$ & \cref{sec:congestion_games} \\
		\hline
        Graphical & $\max(\log_2 n,n/d)-O(\log\log n)$ & $2n/d+2\log_2 d + 1$ & \cref{sec:graphical_games} \\
        \hline
        Auctions & $\log_2 n-O(\log\log n)$ & $(1+\log_2n)(1+c)+1$ & \cref{sec:auctions} \\
        \hline
	\end{tabular}
	\caption{Summary of results.}
	\label{table:summary}
\end{table*}

Depending on the specific scenario of application, sometimes external constraints may limit the type of games that can be designed. For example, in an auction platform, the algorithm may only use auctions. Therefore, we study different settings of the type of games the algorithm may design. We summarize our results for different settings in \cref{table:summary}.

\textbf{Information-theoretical lower bound.} We first show that any algorithm that learns the coalition structure requires at least $\log_2 n - O(\log\log n)$ rounds regardless of the type of games used, where $n$ is the number of agents. 

\textbf{Normal-form games.} Normal-form games are a general class of games. We consider the case when the platform can design any normal-form game and there are no external constraints on the type of games that can be designed. We show an algorithm, \textit{simultaneous binary search}, that can learn the coalition structure in $\log_2 n + 2$ rounds using \textit{directed prisoner's dilemmas}, a type of designed games for agents to play to elicit coalition structures. The number of rounds matches the lower bound almost exactly. 

\textbf{Congestion games.} Congestion games \cite{rosenthal1973class} are a well-studied type of games that model the allocation of resources, and has demonstrated its importance in real-world transportation system. For CSL in congestion games, we show that with a different type of games -- \textit{directed Braess's paradoxes} -- the simultaneous binary search algorithm can be adapted to learn the coalition structure with congestion games. The number of rounds required is also $\log_2 n + 2$.

\textbf{Graphical games.} Graphical games \cite{kearns2001graphical} are another thoroughly researched class of games. Investigating CSL with graphical games can be useful in contexts like social networks. We consider degree-$d$ graphical games, i.e., each agent has at most $d$ neighbors in the underlying graph. In the negative direction, we show that any algorithm for the problem requires at least $\lceil\frac{n-1}{d}\rceil$ rounds. 
In the positive direction, 
we design an algorithm that couples simultaneous binary search with the idea of \textit{block decomposition}, i.e., to divide the agents into blocks of size $\lfloor\frac{d}{2}\rfloor$.
See the construction of \cref{alg:graphical,fig:decomposition} for details of this technique. Our algorithm learns the coalition structure in $2n/d+2\log_2 d + 1$ rounds. Given the two lower bounds, this algorithm is optimal up to a constant factor.

\textbf{Auctions.} Finally, we study CSL with auctions, in particular, \emph{second-price auctions with personalized reserves}, which are widely adopted in online advertising systems. The format of auctions is much more restrictive and challenging than the previous types of games since there is only one winner in an auction, and it is difficult to elicit information from the other players.
Our algorithm in this case learns the coalition structure in $(1+\log_2n)(1+c)+1$ rounds, where $c$ is the size of the largest coalition. When the coalitions are small, the number of rounds required is close to the lower bound. 

\subsection{Related Work}
\label{subsec:other_related_work}

Our work lies in the general research direction of learning in games~\cite{LearningGameBook}. In particular, our work is closely related to the line of research on Inverse Game Theory, e.g.,~\cite{WZB11, KS15, LFK18, LCM09}, where these previous papers aim to infer the parameters and underlying payoffs of a game by observing the behavior of players. Our work adopts this perspective but seeks to understand the \emph{coalition structure} in games, which is not considered in prior works. 
By strategically and adaptively selecting defender strategies and observing the attacker's best responses,~\cite{balcancommitment15, HaghtalabSecurityGame16,wu2022inverse} show that one can uncover the underlying utility functions driving attacker behavior in Stackelberg security games. 
Our work learns in a similar adaptive manner but focuses on the coalition structure among agents.
In addition, learning coalition structure has a fundamental difference from the above existing literature since the space of all possible coalition structures is exponentially large.

The closest related work is~\cite{xu2023learning}, where they focus on learning coalition structure through a single bit feedback, i.e., the game designed in the learning process only returns one-bit information -- whether a specified strategy profile is a NE or not. In this work, we strictly generalize the model proposed by~\cite{xu2023learning} that allows us to observe multiple-bit feedback rather than just binary feedback. This flexibility strictly increases the set of possible observations and thus the difficulty of designing the games. Moreover,~\cite{pmlr-v180-bonjour22a, Mazrooei_Archibald_Bowling_2013} provide general approaches to detect \emph{one} (collusion) coalition in multi-agent games, whereas, our work aims to find the entire
coalition structure in games.

There is also a line of research on \emph{learning with revealed preference feedback}, e.g.,~\cite{Beigman06, Zadimoghaddam12, Balcan14, BlumLearnOptCommit14, Amin15, RothWatchandLearn16}, in which they target to predict the strategic agents' behavior or optimize the profit of the decision maker, through the revealed preference feedback from strategic agents. In this paper, we assume that we can observe the best response feedback of the agents for the games, where each game is carefully designed to elicit information on the coalition structure within a set of agents. Loosely related work includes no-regret learning in games~\cite{Cesa-Bianchi_Lugosi_2006,daskalakis2021near}. Indeed, if there is a cost to perform one game to learn coalition structure, our learning process achieves \emph{constant} regret given the number of rounds needed to identify all coalitions only depends on the number of agents and the structure of the underlying games.

\section{Notations and Preliminaries}
\label{sec:preliminaries}

Consider a set $N = \{1,2,\dots,n\}$ of $n$ strategic agents. A coalition $S\subseteq N$ is a nonempty subset of the agents, in which the agents can coordinate with each other in games. We assume there is a coalition structure $\S^* = \{{S_1}^*,{S_2}^*,\dots,{S_m}^*\}$ among these agents, which is a set partition of $N$, and each agent in $N$ belongs to exactly one of these mutually disjoint coalitions. We use $[i]_{\S^*}$ to denote the coalition containing agent $i$ under $\S^*$.
$\S^*$ is not public information, but agents in each coalition know the other agents in the same coalition.
In the Coalition Structure Learning (CSL) problem, our goal is to actively design a sequence of games for the agents to play, such that observations in these games can be used to learn the underlying coalition structure $\S^*$. Specifically, we interact with the agents in a sequence of rounds. Each round, we design a game $\G$ along with a strategy profile $\Sigma=(\sigma_1,\sigma_2,\dots,\sigma_n)$ in $\G$ for the agents, and we make an observation $\O(\G, \Sigma)$ about the strategy profile $\Sigma$ in $\G$. The task is to learn the underlying coalition structure $\S^*$.

\paragraph{Behavior model.} We assume that the agents are rational and strategic. When playing a game $\G$, agents in the same coalition under $\S^*$ can coordinate with each other to choose actions and will share rewards after the game. Thus, they effectively play as a joint player whose action space and utility are, respectively, the Cartesian product of the action spaces and the sum of utilities of the agents in that coalition. Since the agents are not aware of the underlying coalition structure beyond their own coalitions, the information in the game is \emph{incomplete}. We do not assume the agents' full behavior model, but focus on the agents' deviation decisions from the specified strategy profile.


\paragraph{Multiple-bit observation oracle.} In this paper, we consider the multiple-bit observation oracle as opposed to the single-bit observation oracle considered in \cite{xu2023learning}. The observation $\O(\G,\Sigma)$ is an $n$-dimensional binary vector, where the $i$-th bit $\O_i(\G,\Sigma)\in\{\mathrm{True}, \mathrm{False}\}$ indicates whether or not agent $i$ would like to \emph{deviate} from the specified strategy $\sigma_i$ when the coalition is best responding as a whole. In particular, given a game $\G$ and a strategy profile $\Sigma$, for each coalition $S\in\S^*$, agents in $S$ would first determine whether $\Sigma_S$ is a joint best response to the other agents' specified strategy $\Sigma_{-S}$ in $\G$. If it is, then $\O_i = \mathrm{False}$ for each agent $i \in S$, i.e., they would decide not to deviate. Otherwise, they would compute a possible joint best response $\BR(S,\Sigma)$ to the other agents' specified strategy $\Sigma_{-S}$ in $\G$. Then, they would compute the deviation decision $\O_i=\ind{\BR_i(S,\Sigma)\ne \sigma_i}$ for each agent $i\in S$. The observation $\O(\G,\Sigma)$ is the concatenation of the decisions of all agents in $N$. If there are multiple best responses for a coalition, we assume that the agents in that coalition would choose any of them in the worst case. We will use ``query the observation oracle'' or ``query $(\G,\Sigma)$'' to refer to the process of designing a game $\G$ and a strategy profile $\Sigma$ and receiving the observation $\O(\G,\Sigma)$.

\subsection{Lower Bound of the Number of Rounds}
\label{subsec:lower_bound}

We first show a lower bound on the number of rounds required to solve the Multiple-bit CSL problem. The lower bound is based on an information-theoretical argument similar to \cite{xu2023learning}. Note that this lower bound holds for any type of games that the algorithm may design.

\begin{theorem}
	\label{thm:lower_bound}
	Any algorithm that solves the Multiple-bit CSL problem requires at least $\log_2 n-O(\log\log n)$ rounds of interactions with the agents in the worst case.
\end{theorem}

\begin{proof}
	In each round, the algorithm receives at most $n$ bits of information. As the number of possible set partitions of $N$ is the Bell number $B_n$, to distinguish between them, we need at least $\left\lceil \log_2 B_n\right\rceil=n \log_2 n - O(n\log_2\log_2 n)$ bits of information. The equation follows from the asymptotic expression of Bell number established in \cite{de1981asymptotic}. The theorem then follows.
\end{proof}

\section{Multiple-bit CSL with Normal-form Games}
\label{sec:normal_form_games}

We show in this section an algorithm that solves Multiple-bit CSL with normal-form games in $\log_2 n + 2$ rounds. Note that the number of rounds matches the lower bound in \cref{thm:lower_bound} up to $O(\log \log n)$. The main idea of the algorithm is to use binary search to find another agent in the same coalition for each agent simultaneously. We start by defining game-strategy pairs and a special pair that we call the \textit{directed prisoner's dilemma}, which we use in the algorithm.

\begin{definition}
    \label{def:game-strategy-pair}
    A \textbf{game-strategy pair} $(\G, \Sigma)$ is a $n$-player normal-form game along with a pure strategy profile $\Sigma$ in $\G$. Here, $\Sigma$ is called the \textbf{specified strategy profile} of $(\G, \Sigma)$.
\end{definition}

\begin{definition}
    \label{def:directed_prisoner_dilemma}
	For $i\in N, j\in N$, a \textbf{directed prisoner's dilemma} $\P(i,j)$ is a game-strategy pair $(\G, \Sigma)$. In game $\G$, agent $j$ has two actions $\{\textup{C},\textup{D}\}$, and all other agents only have one action $\{\textup{D}\}$. If agent $j$ plays $\textup{C}$, then $i$ and $j$ receive utility $2$ and $-1$ respectively. Otherwise they both receive $0$ utility. All other agents always receive $0$ utility. $\Sigma$ is the pure strategy profile where all agents play $\textup{D}$.
\end{definition}

\begin{figure}[htbp]
	\centering
	\begin{equation*}
		\begin{array}{|c|c|c|}
			\hline & \textup{C} & \textup{D} \\
			\hline \textup{D} & (2,-1) & (0,0) \\
			\hline
		\end{array}
	\end{equation*}
	\caption{The payoff table of agents $i$ and $j$ in the directed prisoner's dilemma $\P(i,j)$. Agent $i$ is the row player and agent $j$ is the column player. The others only have one action and are not shown in the table.}
	\label{fig:directed_prisoner_dilemma}
\end{figure}

Given the construction of \textbf{directed prisoner's dilemma}, we have the following observation.

\begin{lemma}
\label{lem:directed_prisoner_dilemma}
Let $i,j\in N$ and $\mathbf O = \O(\P(i,j))$. Then, $O_j = \mathrm{True}$ if and only if agent $i$ and agent $j$ are in the same coalition under $\S^*$. Moreover, $O_k = \mathrm{False}$ for $k\in N\setminus\{j\}$.
\end{lemma}

\begin{proof}
	If agent $i$ and agent $j$ are in the same coalition under $\S^*$, then playing $\textup{C}$ is a dominant strategy for agent $j$'s coalition. Otherwise, playing $\textup{D}$ is a dominant strategy. Therefore, $O_j = \mathrm{True}$ if and only if agent $i$ and agent $j$ are in the same coalition under $\S^*$. Moreover, since all agents other than $j$ only have one action $\{\textup{D}\}$, $O_k = \mathrm{False}$ for $k\in N\setminus\{j\}$.
\end{proof}

\cref{lem:directed_prisoner_dilemma} shows that we can use directed prisoner's dilemmas to determine whether two agents are in the same coalition under $\S^*$. However, querying a single directed prisoner's dilemma is not efficient as we only get $1$ bit of information. To utilize the information from the observation oracle more efficiently, we will need to consider the \textit{product} of multiple directed prisoner's dilemmas.

\begin{definition}[\cite{xu2023learning}]
	\label{def:product_of_normal_form_games}
    Let $(\G_1,\Sigma_1),(\G_2, \Sigma_2)$ be two game-strategy pairs where $A_{x,i},u_{x,i}$ are the action set and utility function of agent $i$ in $\G_x$ respectively for $x\in\{1,2\}$. Let $\Sigma_1 = (\sigma_{1,i})_{i\in N}$ and $\Sigma_2 = (\sigma_{2,i})_{i\in N}$. 
    The \textbf{product} of $(\G_1, \Sigma_1)$ and $(\G_2, \Sigma_2)$ is a game-strategy pair $(\G_p, \Sigma_p)$. Here, $\G_p$ is a normal-form game with action set $A_{1,i}\times A_{2,i}$ and utility function $u_{1,i} + u_{2,i}$ for each $i\in N$. $\Sigma_{p,i} = (\sigma_{1,i},\sigma_{2,i})_{i\in N}$. We denote the product of $(\G_1, \Sigma_1)$ and $(\G_2, \Sigma_2)$ as $(\G_1, \Sigma_1)\times(\G_2, \Sigma_2)$ or $\prod_{x=1}^{2}(\G_x, \Sigma_x)$. The game-strategy pairs $(\G_1, \Sigma_1),(\G_2, \Sigma_2)$ are called the \textbf{factors} of $(\G_p, \Sigma_p)$.
\end{definition}


\begin{lemma}
	\label{lem:product_of_directed_prisoner_dilemma}
	Let $\{(i_x,j_x)\mid x \in \{1,2,\dots,k\}\}$ be $k$ pairs of agents and $\mathbf O = \O(\prod_{x=1}^{k} \P(i_x,j_x))$. Then, $O_j = \mathrm{True}$ if and only if there exists $x\in\{1,2,\dots,k\}\textbf{ such that }j_x = j$ and $i_x\in[j]_{\S^*}$.
\end{lemma}

\begin{proof}
	By \cref{def:product_of_normal_form_games}, playing the product game is equivalent to separately playing each factor game, and summing up the resulting utilities of each agent. Therefore, agent $j$ decides to deviate in the product game $\prod_{x=1}^{k} \P(i_x,j_x)$ if and only if agent $j$ decides to deviate in at least one factor game. The result then follows from \cref{lem:directed_prisoner_dilemma}.
\end{proof}

A graphical interpretation of \cref{lem:product_of_directed_prisoner_dilemma} is as follows. When we query the game-strategy pair $\prod_{x=1}^{k} \P(i_x,j_x)$, consider a directed graph where each vertex represents an agent and each edge $(j_x,i_x)$, $x \in \{1,2,\dots,k\}$ represents $\P(i_x,j_x)$. The observation tells us for each vertex if any of its outgoing neighbors 
in this constructed graph is in the same coalition. With \cref{lem:product_of_directed_prisoner_dilemma}, we are ready for our \cref{alg:normal_form} for Multiple-bit CSL with normal-form games. 

\IncMargin{1.0em}
\begin{algorithm}[ht]

	\Indmm\Indmm
	\KwIn{The number of agents $n$ and the multiple-bit observation oracle $\O$.}
	\KwOut{A coalition structure $\S$ of the agents.}
	\Indpp\Indpp
	
	Query $\prod_{i\in N,j\in N, i<j}\P(i,j)$ and observe $\mathbf O$\;
	Let $T_j \gets \{1,2,\dots,j-1\}\textbf{ if } O_j=\mathrm{True}\textbf{ else }\varnothing$ $\textbf{ for each }j\in N$\;
	\While{$\exists j \in N\textbf{ such that }|T_j|\geq 2$}{
		Let $L_j \gets \{\text{the smallest $\lfloor \frac{|T_j|}2\rfloor$ elements in $T_j$}\}$\textbf{ for each }$j\in N$\;
		Let $R_j \gets T_j \setminus L_j$\textbf{ for each }$j\in N$\;
		Query $\prod_{j\in N,i\in L_j}\P(i,j)$ and observe $\mathbf O$\;
		Let $T_j \gets L_j\textbf{ if } O_j=\mathrm{True}\textbf{ else } R_j$ $\textbf{ for each }j\in N$\;
	}
	Let $\S\gets\{\{1\},\{2\},\dots,\{n\}\}$\;
	\For{$j\in N\textbf{ such that }|T_j|\ne \varnothing$}{
		Let $i\gets\text{the only element in $T_j$}$\;
		Merge $[i]_\S$ and $[j]_\S$ in $\S$\;
	}
	\Return{$\S$}\;
	\caption{Simultaneous Binary Search}
	\label{alg:normal_form}
\end{algorithm}
\DecMargin{1.0em}

Intuitively, given an agent $j$, if we want to find another agent $i\ (i<j)$ in the same coalition with $j$, we can query a game-strategy pair $\prod_{k\in T_j} \P(k,j)$ where $T_j =\{1,2,\ldots,j-1\}$ initially. According to \cref{lem:product_of_directed_prisoner_dilemma}, we will know whether there is such an agent in $T_j$ from the $j$-th bit of the observation. If there is such an agent, we can then bisect the set $T_j$ and apply binary search to find this agent. Moreover, this binary search can be done simultaneously for different agents. For example, given two agents $j$ and $j'$, we can construct the game-strategy pairs $\prod_{k\in T_{j}} \P(k,j)\times\prod_{k\in T_{j'}} \P(k,j')$ and apply binary search at the same time on $T_{j}$ and $T_{j'}$ using the $j$-th and $j'$-th bits of the observation. This is because the $j$-th bit in the observation vector only depends on whether there exists another agent in $T_{j}$ in the same coalition with agent $j$, and the same argument applies to agent $j'$. 

\cref{alg:normal_form} leverages this intuition and works in two stages.
In the first stage (Lines 1 to 7), we try to find for each agent $j$ the smallest index of the agents in $[j]_{\S^*}$. To do this, we first query $\prod_{i\in N,j\in N, i<j}\P(i,j)$ (Line 1). In this game, agent $j$ has the incentive to deviate only when another agent $i$ with a smaller index than $j$ is in the same coalition with $j$. Therefore, using the $j$-th bit from the observation, we can determine whether agent $j$ is the agent with the smallest index in $[j]_{\S^*}$ (Line 2). For agents who are not, we then use binary search to locate the smallest indexed agents in their coalitions simultaneously (Lines 3 to 6).
In the second stage (Lines 8 to 11), it merges each coalition in $\S^*$ together according to the result of the first stage.

We provide an example of \cref{alg:normal_form} in \cref{fig:binary}.

\begin{figure*}[!t]
    \centering
    
    \begin{subfigure}[b]{0.32\textwidth}
        \centering
        \begin{tikzpicture}[scale=1]
            \node[circle, draw] (1) at (0,2) {1};
            \node[circle, draw] (3) at (2,2) {3};
            \node[circle, draw] (2) at (0,0) {2};
            \node[circle, draw] (4) at (2,0) {4};

            \draw[->, >=latex, line width=0.5pt] (2) -- (1);
            \draw[->, >=latex, line width=0.5pt] (3) -- (1);
            \draw[->, >=latex, line width=0.5pt] (3) -- (2);
            \draw[->, >=latex, line width=0.5pt] (4) -- (1);
            \draw[->, >=latex, line width=0.5pt] (4) -- (2);
            \draw[->, >=latex, line width=0.5pt] (4) -- (3);
        \end{tikzpicture}
        \caption{First query}
		\label{subfigure:binary_a}
    \end{subfigure}
    \begin{subfigure}[b]{0.32\textwidth}
        \centering
        \begin{tikzpicture}[scale=1]
            \node[circle, draw] (1) at (0,2) {1};
            \node[circle, draw] (3) at (2,2) {3};
            \node[circle, draw] (2) at (0,0) {2};
            \node[circle, draw] (4) at (2,0) {4};

            \node at (0,1) {$T_3$};
            \node at (1.5,1.2) {$T_4$};

            \draw[rounded corners, dashed] ($(1.north west)+(-0.2,0.2)$) rectangle ($(2.south east)+(0.2,-0.2)$);
			\draw[rounded corners, dashed] ($(3.north east)+(0.35,0.35)$) -- ($(3.south east)+(0.35,-0.35)$) -- ($(1.south east)+(0.35,-0.35)$) -- ($(2.south east)+(0.35,-0.35)$) -- ($(2.south west)+(-0.35,-0.35)$) -- ($(1.north west)+(-0.35,0.35)$) -- cycle;

            \draw[->, >=latex, line width=0.5pt] (3) -- (1);
            \draw[->, >=latex, line width=0.5pt] (4) -- (1);
        \end{tikzpicture}
        \caption{Second query}
		\label{subfigure:binary_b}
    \end{subfigure}
    \begin{subfigure}[b]{0.32\textwidth}
        \centering
        \begin{tikzpicture}[scale=1]
            \node[circle, draw] (1) at (0,2) {1};
            \node[circle, draw] (3) at (2,2) {3};
            \node[circle, draw] (2) at (0,0) {2};
            \node[circle, draw] (4) at (2,0) {4};

            \node at (0.7,0) {$T_3$};
            \node at (0.7,2) {$T_4$};

            \draw[rounded corners, dashed] ($(1.north west)+(-0.2,0.2)$) rectangle ($(1.south east)+(0.2,-0.2)$);
            \draw[rounded corners, dashed] ($(2.north west)+(-0.2,0.2)$) rectangle ($(2.south east)+(0.2,-0.2)$);
        \end{tikzpicture}
        \caption{Final result}
		\label{subfigure:binary_c}
    \end{subfigure}
    \caption{Example execution of \cref{alg:normal_form} when $\S^* = \{\{1,4\}, \{2, 3\}\}$. The vertices represent the agents, the edges represent the directed prisoner's dilemmas that the algorithm queries each time, and the dashed rectangles represent the sets $T_j$. In the first query, the algorithm queries $\prod_{i\in N,j\in N, i<j}\P(i,j)$ (Line 1) as shown in (a). Using the observations, the algorithm sets $T_1 = T_2 = \varnothing $, $T_3 = \{1, 2\}$, and $T_4 = \{1, 2, 3\}$ (Line 2). In the second query, $T_3$ is bisected into $L_3 = \{1\},R_3 = \{2\}$, $T_4$ is bisected into $L_4 = \{1\},R_4 = \{2,3\}$, and the algorithm queries $\prod_{j\in N,i\in L_j}\P(i,j)$ (Lines 3 to 6) as shown in (b). Using the observations, the algorithm sets $T_3 = L_3 = \{1\}$ and $T_4 = R_4 = \{2,3\}$ (Line 7) as shown in (c). Finally, the algorithm merges $\{1\}$ and $\{4\}$, and $\{2\}$ and $\{3\}$ together to recover the coalition structure (Lines 8 to 11).}
	\label{fig:binary}
\end{figure*}
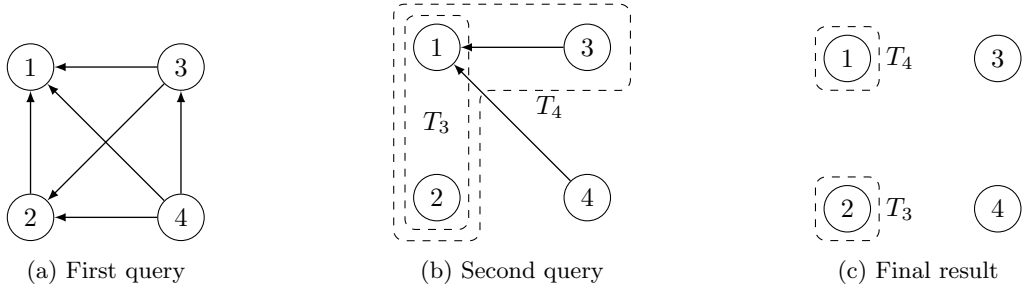

\begin{restatable}{theorem}{normal}
	\label{thm:normal_form}
	Algorithm \ref{alg:normal_form} solves Multiple-bit CSL with normal-form games in $\log_2 n + 2$ rounds.
\end{restatable}


\begin{proof}
	We first show the correctness of the algorithm. To do this, we will show two claims: (i) after Lines 1 to 7, for each $j\in N$, if $j$ is the agent with the smallest index in $[j]_{S^*}$, $T_j=\varnothing$, otherwise, $T_j$ contains only the smallest index of the agents in $[j]_{S^*}$; (ii) after Lines 8 to 11, $\S = \S^*$.
	
	For (i), if $j$ is the agent with the smallest index in $[j]_{S^*}$, then by \cref{lem:product_of_directed_prisoner_dilemma}, $O_j = \mathrm{False}$, and $T_j = \varnothing$ on Line 2. Throughout the algorithm, $T_j$ remains $\varnothing$. Otherwise, $T_j$ contains the smallest index of the agents in $[j]_{S^*}$ after Line 2. In the while loop (Line 3 to 7), $T_j$ is updated to $L_j$ if $O_j = \mathrm{True}$, and to $R_j$ otherwise. Since $L_j$ contains the smallest $\lfloor|T_j|/2\rfloor$ elements in $T_j$, by \cref{lem:product_of_directed_prisoner_dilemma}, $O_j = \mathrm{True}$ if and only if the smallest index of the agents in $[j]_{S^*}$ is in $L_j$. Therefore, after one iteration of the loop, $T_j$ still contains only the smallest index of the agents in $[j]_{S^*}$, while the size of $T_j$ is halved. The loop terminates after $|T_j|=1$, and $T_j$ contains only the smallest index of the agents in $[j]_{S^*}$.

	For (ii), since every agent $j$ is either the agent with the smallest index in $[j]_{S^*}$ or merged with that agent, each coalition is merged in $\S$ after Lines 8 to 11. Thus $\S$ becomes the same as $\S^*$.

	Next, we show the complexity of the algorithm. The while loop in Lines 3 to 7 runs at most $\lceil\log_2 n\rceil$ times, and each iteration requires $1$ query. Together with the query on Line 1, the total number of queries is at most $\lceil\log_2 n\rceil + 1 \leq \log_2 n + 2$.
\end{proof}

\paragraph{Multiple-bit CSL with Other Types of Games.} A natural next step after this section is to consider the case where the algorithm can only design games of a specific type. We will show that the idea of simultaneous binary search from \cref{alg:normal_form} can be adapted to solve the Multiple-bit CSL problem with congestion games (\cref{sec:congestion_games}), graphical games (\cref{sec:graphical_games}), or auctions (\cref{sec:auctions}). The adaptation to congestion games is natural once we construct a new game-strategy pair that is similar in spirit to the directed prisoner's dilemma. However, the adaptation to graphical games and auctions requires more innovations in algorithm design.

\section{Multiple-bit CSL with Auctions}
\label{sec:auctions}

We consider in this section the case of Multiple-bit CSL with auctions. The auction format we consider is theoretically well-studied \textit{second-price auctions with personalized reserves} \cite{paes2016field}, which is widely used in practice. We first formally define this format.

\begin{definition}
    \label{def:auction}
    A \textbf{second-price auction with personalized reserves} $\G$ is a tuple $(N, \v, \r)$ where $N$ is the set of agents, $\v = (v_1,v_2,\dots,v_n)$ is the valuation vector, and $\r = (r_1,r_2,\dots,r_n)$ is the reserve price vector. The auction proceeds as follows: Each agent $i$ submits a bid $\sigma_i \in \mathbb R_{\geq 0}$ and the agent with the highest bid wins the auction with ties broken at uniform random. If the winner $i$ bids greater than $r_i$, then the winner gets the item at a price equal to the maximum between the second highest bid and $r_i$. The item can be reallocated among $[i]_{\S^*}$. Otherwise, the item is not allocated. 
\end{definition}

Our algorithm for the Multiple-bit CSL problem with auctions works in $(1 + \log_2 n)(1 + c) + 1$ rounds, where $c$ is size of the largest coalition in $\S^*$. To introduce the algorithm, we first define the game-strategy pairs that we will use in the algorithm, the \textit{auction game gadgets}, where we will specify the valuation and reserve prices of every agents in each partition bucket.

\begin{definition}
	\label{def:auction_gadget}
	Let $\{X,Y,Z\}$ be a partition of $N$, i.e., $X \cap Y = X \cap Z = Y \cap Z =\varnothing$ and $X \cup Y \cup Z = N$. An \textbf{auction gadget} $\A(X, Y, Z)$ is a game-strategy pair $(\G, \Sigma)$ where $\G = (N, \v, \r)$ is a second-price auction with personalized reserves, and $\Sigma$ is a specified bid vector for the agents, where
	\begin{equation*}
		v_i = \left\{\begin{array}{ll}
            1 & (i \in X) \\
            0 & (i \in Y) \\
			0 & (i \in Z),
        \end{array}\right.
		r_i = \left\{\begin{array}{ll}
			1 & (i \in X) \\
			0 & (i \in Y) \\
			1 & (i \in Z),
		\end{array}\right.\text{and}
        \ \ \ 
		\sigma_i = 0
	\end{equation*}
\end{definition}

Note that in the above auction game gadgets, we do not require each agent to know all the other agents' valuations so that the valuation of each agent is still \emph{private} across strategic agents. Only agents in the same coalition know each other's true valuation. 

\begin{lemma}
	\label{lem:auction_gadget}
	For $\{X,Y,Z\}$ that is a partition of $N$, let $\mathbf O = \O(\A(X, Y, Z))$. For each coalition $S \in \S^*$, if $S \cap X \ne \varnothing$ and $S \cap Y \ne \varnothing$, then, there is exactly one agent $i \in S \cap Y$ with $O_i = \mathrm{True}$ and $O_j = \mathrm{False}$ for all $j \in S \setminus \{i\}$. Otherwise, $O_i = \mathrm{False}$ for all $i \in S$.
\end{lemma}

\begin{proof}
	For each coalition $S \in \S^*$. We first consider the case where $S \cap X \ne \varnothing$ and $S \cap Y \ne \varnothing$. Let agent $i \in S \cap X$ and agent $j \in S \cap Y$. If the coalition follows the specified strategy profile $\Sigma_S$, then the item will not be allocated, so the utility of the agents in $S$ is $0$. However, If agent $j$ raises its bid to $\sigma'_j > 0$, then it becomes the winner of the auction, and the item will be allocated to $j$ with price $\max\{second\ highest\ bid, r_j\} = 0$. The item can then be reallocated to agent $i$, giving utility $1$ to the agents in $S$. Thus, $\Sigma_S$ is not a joint best response to $\Sigma_{-S}$. Moreover, for agents in $S$, the only joint best response is to let exactly one agent in $S \cap Y$ raise its bid to positive, and the other agents in $S$ keep their bids at $0$. Thus, the observation $\mathbf O$ is as described in the lemma.
	
	Otherwise, either $S \cap X = \varnothing$ or $S \cap Y = \varnothing$. If $S \cap X = \varnothing$, then no one in $S$ has a positive valuation for the item, so the utility of the agents in $S$ cannot be greater than $0$. If $S \cap Y = \varnothing$, then no one in $S$ has a reserve price that is smaller than $1$, which is the largest possible valuation. Thus, the utility of the agents in $S$ cannot be greater than $0$, either. In both cases, keeping the bids at $0$ is a joint best response to $\Sigma_{-S}$, so the observation $\mathbf O$ is as described in the lemma.
\end{proof}


The auction gadgets are important building blocks for our algorithm that works in $(1 + \log_2 n)(1 + c) + 1$ rounds. To help demonstrate how they can be used to solve Multiple-bit CSL, we present a simpler algorithm that works in $n - 1$ rounds using the auction gadgets below. 

\IncMargin{1.0em}
\begin{algorithm}[ht]

	\Indmm\Indmm
	\KwIn{The number of agents $n$ and the multiple-bit observation oracle $\O$.}
	\KwOut{A coalition structure $\S$ of the agents.}
	\Indpp\Indpp
	
	Let $T \gets N$\;
	Let $\S\gets\{\{1\},\{2\},\dots,\{n\}\}$\;
	\While{$|T| \geq 2$}{
		Let $i\gets \text{the first element in $T$}$\;
		\While{$|T| \geq 2$\textbf{ and } $i \in T$}{
			Let $X \gets \{i\}, Y \gets T \setminus \{i\}$ and $Z \gets N \setminus S$\;
			Query $\A(X,Y,Z)$ and observe $\mathbf O$\;
			\uIf{$\exists j$ \textbf{such that} $O_j = \mathrm{True}$}{
				Merge $[i]_\S$ and $[j]_\S$ in $\S$\;
				$T \gets T \setminus \{j\}$\;
			}\Else{
				$T \gets T \setminus \{i\}$\;
			}
		}
	}
	\Return{$\S$}\;
	\caption{Iterative Location with Auctions}
	\label{alg:auction_linear}
\end{algorithm}
\DecMargin{1.0em}

To provide some intuitions, consider an agent $i \in N$. If we want to find another agent in the same coalition as $i$, we can query the auction gadget $\A(\{i\}, N \setminus \{i\}, \varnothing)$. According to \cref{lem:auction_gadget}, if the observation $\mathbf O$ has $O_j = \mathrm{True}$ for some $j \in N \setminus \{i\}$, then $i$ and $j$ are in the same coalition under $\S^*$. Otherwise, $i$ is in a coalition by itself. Then, suppose we have found an agent $j$ in $i$'s coalition, and we would like to find a different agent in that coalition. We can query the auction gadget $\A(\{i\}, N \setminus \{i,j\}, \{j\})$. By moving agent $j$ to the set $Z$, we ensure that $j$ does not have the incentive to deviate, and either another agent in the same coalition as $i$ will have $O_k = \mathrm{True}$ for some $k \in N \setminus \{i,j\}$, or we can conclude that $i$ and $j$ are in a coalition by themselves. We can then repeat this process until we have found all the agents in the same coalition as $i$. After that, we can move $i$ to the set $Z$ and proceed with the next agent in $N$. In this way, we can recover the coalition structure $\S^*$. Moreover, since each time we query an auction gadget, we move one agent to set $Z$, and the algorithm terminates when $|Z| = n - 1$, it uses $n - 1$ queries in total.

\begin{restatable}{theorem}{auctionlinear}
	\label{thm:auction_linear}
	\cref{alg:auction_linear} solves the Multiple-bit CSL problem with auctions in $n - 1$ rounds.
\end{restatable}

We defer the proof to \cref{app:missing_parts}. Other than the auction gadgets, we will also use another game-strategy pair in the algorithm. We prove the following lemma about the observations obtained by querying this game-strategy pair.

\begin{lemma}
	\label{lem:auction_first_move}
	Let $\mathbf O = \O((N, \v = \mathbf 1, \r = \mathbf 0), \Sigma = \mathbf 0)$, where $\mathbf 0, \mathbf 1$ are vectors consisting purely of $0$s and $1$s respectively. For each coalition $S \in \S^*$, exactly one agent $i \in S$ has $O_i = \mathrm{True}$.
\end{lemma}

\begin{proof}
	For each coalition $S \in \S^*$, if any agent $i$ in $S$ raises its bid to positive, then the item will be allocated to $i$ with price $0$, giving utility $1$ to the agents in $S$. Thus, $\Sigma_S$ is not a joint best response to $\Sigma_{-S}$. Moreover, the only joint best response for agents in $S$ is to let exactly one agent in $S$ raise its bid to positive, and the other agents in $S$ keep their bids at $0$. The lemma then follows.
\end{proof}

Then, we are ready to present our algorithm that works in $(1 + \log_2 n)(1 + c) + 1$ rounds for Multiple-bit CSL with auctions. The algorithm is shown in \cref{alg:auction_sublinear}.

\IncMargin{1.0em}
\begin{algorithm}[ht]

	\Indmm\Indmm
	\KwIn{The number of agents $n$ and the multiple-bit observation oracle $\O$.}
	\KwOut{A coalition structure $\S$ of the agents.}
	\Indpp\Indpp
	
	Query $((N, \v = \mathbf 1, \r = \mathbf 0), \Sigma = \mathbf 0)$ and observe $\mathbf O$\;
	Let $T_x \gets \{i \mid O_i = \mathrm{True}\}$ and $T_y \gets N \setminus T_x$\;
	\For{$b \in \{0,1,\dots,\lfloor \log_2 n \rfloor\}$}{
		Let $X \gets \{i \mid i \in T_x\textbf{ and }$ $\text{the $b$-th lowest binary bit of $i$ is $1$}\}$\;
		Let $T_{\mathrm{False}} \gets T_y$\;
		\Repeat{$T_{\mathrm{True}} = \varnothing$}{
			Let $Y \gets T_{\mathrm{False}}$ and $Z \gets N \setminus (X \cup Y)$\;
			Query $\A(X,Y,Z)$ and observe $\mathbf O$\;
			Let $T_{\mathrm{True}} \gets \{i \mid i \in T_{\mathrm{False}}\textbf{ and }O_i = \mathrm{True}\}$\;
			$T_{\mathrm{False}} \gets T_{\mathrm{False}} \setminus T_{\mathrm{True}}$\;
		}
		Let $O^{(b)}_i \gets \ind{i \not\in T_{\mathrm{False}}}\textbf{ for each }i \in T_y$\;
	}
	Let $\S\gets\{\{1\},\{2\},\dots,\{n\}\}$\;
	\For{$i\in T_y$}{
		Let $j \gets \sum_{b=0}^{\lfloor\log_2 n\rfloor}2^b\cdot O^{(b)}_i$\;
		Merge $[i]_\S$ and $[j]_\S$ in $\S$\;
	}
	\Return{$\S$}\;
	\caption{Bitwise Search with Auctions}
	\label{alg:auction_sublinear}
\end{algorithm}
\DecMargin{1.0em}

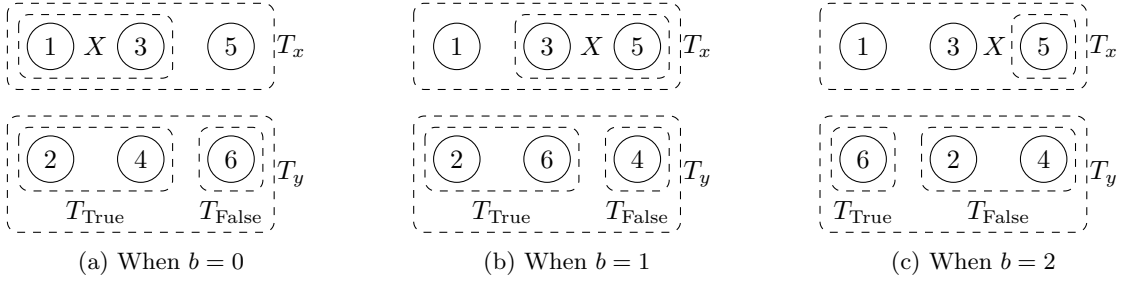
\begin{figure*}[!t]
    \centering
    
    \begin{subfigure}[b]{0.32\textwidth}
        \centering
        \begin{tikzpicture}[scale=1]
            \node[circle, draw] (1) at (0,1.5) {1};
            \node[circle, draw] (3) at (1.2,1.5) {3};
            \node[circle, draw] (5) at (2.4,1.5) {5};
            \node[circle, draw] (2) at (0,0) {2};
            \node[circle, draw] (4) at (1.2,0) {4};
            \node[circle, draw] (6) at (2.4,0) {6};

            \draw[rounded corners, dashed] ($(1.north west)+(-0.2,0.2)$) rectangle ($(3.south east)+(0.2,-0.2)$);
            \draw[rounded corners, dashed] ($(1.north west)+(-0.35,0.35)$) rectangle ($(5.south east)+(0.35,-0.35)$);
            \draw[rounded corners, dashed] ($(2.north west)+(-0.2,0.2)$) rectangle ($(4.south east)+(0.2,-0.2)$);
            \draw[rounded corners, dashed] ($(6.north west)+(-0.2,0.2)$) rectangle ($(6.south east)+(0.2,-0.2)$);
            \draw[rounded corners, dashed] ($(2.north west)+(-0.35,0.35)$) rectangle ($(6.south east)+(0.35,-0.75)$);

            \node at (3.2,1.5) {$T_x$};
            \node at (3.2,-0.2) {$T_y$};
            \node at (0.6,1.5) {$X$};
            \node at (0.6,-0.7) {$T_{\mathrm{True}}$};
            \node at (2.4,-0.7) {$T_{\mathrm{False}}$};
        \end{tikzpicture}
        \caption{When $b = 0$}
		\label{subfigure:bitwise_a}
    \end{subfigure}
    \begin{subfigure}[b]{0.32\textwidth}
        \centering
        \begin{tikzpicture}[scale=1]
            \node[circle, draw] (1) at (0,1.5) {1};
            \node[circle, draw] (3) at (1.2,1.5) {3};
            \node[circle, draw] (5) at (2.4,1.5) {5};
            \node[circle, draw] (2) at (0,0) {2};
            \node[circle, draw] (6) at (1.2,0) {6};
            \node[circle, draw] (4) at (2.4,0) {4};

            \draw[rounded corners, dashed] ($(3.north west)+(-0.2,0.2)$) rectangle ($(5.south east)+(0.2,-0.2)$);
            \draw[rounded corners, dashed] ($(1.north west)+(-0.35,0.35)$) rectangle ($(5.south east)+(0.35,-0.35)$);
            \draw[rounded corners, dashed] ($(2.north west)+(-0.2,0.2)$) rectangle ($(6.south east)+(0.2,-0.2)$);
            \draw[rounded corners, dashed] ($(4.north west)+(-0.2,0.2)$) rectangle ($(4.south east)+(0.2,-0.2)$);
            \draw[rounded corners, dashed] ($(2.north west)+(-0.35,0.35)$) rectangle ($(4.south east)+(0.35,-0.75)$);

            \node at (3.2,1.5) {$T_x$};
            \node at (3.2,-0.2) {$T_y$};
            \node at (1.8,1.5) {$X$};
            \node at (0.6,-0.7) {$T_{\mathrm{True}}$};
            \node at (2.4,-0.7) {$T_{\mathrm{False}}$};
        \end{tikzpicture}
        \caption{When $b = 1$}
		\label{subfigure:bitwise_b}
    \end{subfigure}
    \begin{subfigure}[b]{0.32\textwidth}
        \centering
        \begin{tikzpicture}[scale=1]
            \node[circle, draw] (1) at (0,1.5) {1};
            \node[circle, draw] (3) at (1.2,1.5) {3};
            \node[circle, draw] (5) at (2.4,1.5) {5};
            \node[circle, draw] (6) at (0,0) {6};
            \node[circle, draw] (2) at (1.2,0) {2};
            \node[circle, draw] (4) at (2.4,0) {4};

            \draw[rounded corners, dashed] ($(5.north west)+(-0.2,0.2)$) rectangle ($(5.south east)+(0.2,-0.2)$);
            \draw[rounded corners, dashed] ($(1.north west)+(-0.35,0.35)$) rectangle ($(5.south east)+(0.35,-0.35)$);
            \draw[rounded corners, dashed] ($(2.north west)+(-0.2,0.2)$) rectangle ($(4.south east)+(0.2,-0.2)$);
            \draw[rounded corners, dashed] ($(6.north west)+(-0.2,0.2)$) rectangle ($(6.south east)+(0.2,-0.2)$);
            \draw[rounded corners, dashed] ($(6.north west)+(-0.35,0.35)$) rectangle ($(4.south east)+(0.35,-0.75)$);

            \node at (3.2,1.5) {$T_x$};
            \node at (3.2,-0.2) {$T_y$};
            \node at (1.75,1.5) {$X$};
            \node at (0.0,-0.7) {$T_{\mathrm{True}}$};
            \node at (1.8,-0.7) {$T_{\mathrm{False}}$};
        \end{tikzpicture}
        \caption{When $b = 2$}
		\label{subfigure:bitwise_c}
    \end{subfigure}
    
    \caption{Example execution of the bitwise search (Lines 1 to 12) in \cref{alg:auction_sublinear} when $\S^* = \{\{1,4\},\{2,3\},\{5,6\}\}$. The vertices represent the agents, and the dashed lines represent the sets used in the algorithm. Using the first query, the algorithm identifies one agent in each coalition and groups them as $T_x = \{1, 3, 5\}$. The rest are $T_y = \{2, 4, 6\}$ (Lines 1 to 2). Then, as shown above, for each $b \in \{0,1,2\}$, the algorithm picks the set of agents in $T_x$ with the $b$-th lowest binary bit as $1$ as $X$, and partitions $T_y$ into $T_{\mathrm{True}}$, each of which is cooperating with some agent in $X$, and $T_{\mathrm{False}}$, each of which is not cooperating with any agent in $X$ (Lines 3 to 12).}
	\label{fig:bitwise}
\end{figure*}

The algorithm starts by querying the game-strategy pair $((N, \v = \mathbf 1, \r = \mathbf 0), \Sigma = \mathbf 0)$ (Line 1). Using the obtained observation, by \cref{lem:auction_first_move}, the algorithm identifies exactly one agent in each coalition, groups them into set $T_x$ and the rest of the agents into set $T_y$ (Lines 2). Essentially, set $T_x$ contains a set of agents who are definitely not in the same coalition with each other. Furthermore, for each agent $i$ in $T_y$, there is exactly one ``teammate'' of $i$ in $T_x$. We denote the index of this agent as $\alpha(i)$.

An important idea behind \cref{alg:auction_sublinear} is \textit{bitwise search}, i.e., to determine each binary bit of $\alpha(i)$ separately. To do this, for each $b \in \{0,1,\dots,\lfloor \log_2 n \rfloor\}$, the algorithm picks the set of agents in $T_x$ such that the $b$-th lowest binary bits of their index are $1$ as $X$ (Line 4), and tries to tell which agents in $T_y$ are in the same coalition as some agent in $X$ by querying auction gadgets. Intuitively, if we let $Y = T_y, Z = N \setminus (X \cup Y)$ and query $\A(X, Y, Z)$, by \cref{lem:auction_gadget}, for each coalition that contains an agent in $X$, one of its members in $Y$ will decide to deviate. Thus, we can move these agents to $Z$ like how we did in \cref{alg:auction_linear}, until no agents in $Y$ decide to deviate. The remaining agents in $Y$ are then the agents that are not in the same coalition as any agent in $X$. The algorithm leverages this idea to determine each binary bit of $\alpha(i)$ separately, and stores the results in $O^{(b)}_i$ (Lines 5 to 12). Note that for each $b$, the algorithm uses at most $1 + c$ queries, where $c$ is the size of the largest coalition in $\S^*$. We illustrate the execution of the bitwise search with an example in \cref{fig:bitwise}.

Finally, the algorithm recovers $\alpha(i)$ for each agent $i$ in $T_y$ by summing up the values of the binary bits in $O^{(b)}_i$ (Line 15), and merges $i$ with $\alpha(i)$ accordingly (Line 16).

\begin{restatable}{theorem}{auctionsublinear}
	\label{thm:auction_sublinear}
	\cref{alg:auction_sublinear} solves the Multiple-bit CSL problem with auctions in $(1 + \log_2 n)(1 + c) + 1$ rounds, where $c$ is size of the largest coalition in $\S^*$.
\end{restatable}


\begin{proof}
	We first show the correctness of the algorithm. By \cref{lem:auction_first_move}, after Line 2, $T_x$ contains exactly one agent in each coalition. This means that for each agent $i \in T_y$, there is exactly one agent $j \in T_x$ such that $i$ and $j$ are in the same coalition. We denote this agent $j$ as $\alpha(i)$. Then, we will show two claims: (i) after Lines 3 to 12, for each $b \in \{0,1,\dots,\lfloor \log_2 n \rfloor\}$ and each $i \in T_y$, $O^{(b)}_i = \mathrm{True}$ if and only if the $b$-th lowest binary bit of $\alpha(i)$ is $1$, and (ii) after Lines 13 to 16, $\S = \S^*$.

	For (i), we only need to show that after the repeat loop in Lines 6 to 11, $T_{\mathrm{False}}$ contains exactly the agents $i \in T_y$ such that the $b$-th lowest binary bit of $\alpha(i)$ is $0$. If an agent $i \in T_y$ has the $b$-th lowest binary bit of $\alpha(i)$ as $0$, then $\alpha(i) \not\in X$ after Line 4. Thus, by \cref{lem:auction_gadget}, $O_i$ will never be $\mathrm{True}$ on Line 8, and $i \in T_{\mathrm{False}}$ after the repeat loop. If an agent $i \in T_y$ has the $b$-th lowest binary bit of $\alpha(i)$ as $1$, then $\alpha(i) \in X$ after Line 4. By \cref{lem:auction_gadget}, in each iteration of the repeat loop, one agent $j \in [i]_{\S^*}\cap T_{\mathrm{False}}$ will have $O_j = \mathrm{True}$, and will be removed from $T_{\mathrm{False}}$. Since the repeat loop terminates only when no agent is removed from $T_{\mathrm{False}}$ in one iteration, $i \not\in T_{\mathrm{False}}$ after the loop.

	For (ii), given each binary bit of $\alpha(i)$, the algorithm computes the index of the agent in $\alpha(i)$ in $T_y$ by summing up the values of the binary bits on Line 15. Then, the algorithm merges $i$ with $\alpha(i)$ accordingly on Line 16. The correctness of the algorithm then follows.

	For the complexity, the algorithm uses $1$ query on Line 1, and for each $b \in \{0,1,\dots,\lfloor \log_2 n \rfloor\}$, the algorithm uses at most $1 + c$ queries in the repeat loop in Lines 6 to 11. This shows that the total number of queries used by the algorithm is at most $(1 + \log_2 n)(1 + c) + 1$.
\end{proof}

\section{Conclusion and Future Work}
\label{sec:conclusion}

In this paper, we study the Coalition Structure Learning (CSL) problem under the multiple-bit observation oracle. We consider various settings of the type of games that the algorithm may design, including normal-form games, congestion games, graphical games, and auctions. In each of the settings, we present an algorithm that learns the coalition structure with a sublinear number of rounds. We also complement our algorithmic results with lower bounds on the number of rounds required to learn the coalition structure, demonstrating their optimality in most of the settings. Compared to the existing work, our results significantly reduce the number of rounds required to learn the coalition structure, and greatly improve the potential applicability of the CSL problem in real-world systems. Regarding CSL, there are various different settings for future work that are not covered by this paper. It would be interesting to consider the CSL problem (i) when the underlying coalition structure is dynamic, (ii) when the agents are aware of the algorithm and respond strategically, (iii) when the algorithm can only observe a subset of the agents in each round, and (iv) when the observations are noisy or the agents are boundedly rational.

\newpage

\section*{Acknowledgements}

This work was supported in part by NSF grant IIS-2046640 (CAREER) and NSF IIS-2200410.

\bibliographystyle{unsrtnat}
\bibliography{ref}

\begin{thebibliography}{30}
\providecommand{\natexlab}[1]{#1}
\providecommand{\url}[1]{\texttt{#1}}
\expandafter\ifx\csname urlstyle\endcsname\relax
  \providecommand{\doi}[1]{doi: #1}\else
  \providecommand{\doi}{doi: \begingroup \urlstyle{rm}\Url}\fi

\bibitem[Milgrom(2004)]{milgrom2004putting}
Paul~Robert Milgrom.
\newblock \emph{Putting auction theory to work}.
\newblock Cambridge University Press, 2004.

\bibitem[Hamilton(2019)]{hamilton2019uber}
Isobel~Asher Hamilton.
\newblock Uber drivers are reportedly colluding to trigger `surge' prices because they say the company is not paying them enough.
\newblock \emph{Business Insider}, 2019.

\bibitem[Sweeney(2019)]{Sweeney2019manipulate}
Sam Sweeney.
\newblock Uber, lyft drivers manipulate fares at reagan national causing artificial price surges.
\newblock \emph{WJLA}, 2019.

\bibitem[Dowling(2023)]{Dowling2023drive}
Joshua Dowling.
\newblock How uber drivers trigger fake surge price periods when no delays exist.
\newblock \emph{Drive}, 2023.

\bibitem[Wikipedia(2024)]{Libor}
Wikipedia.
\newblock Libor scandal.
\newblock \url{https://en.wikipedia.org/wiki/Libor_scandal}, 2024.
\newblock Accessed: 2024-08-11.

\bibitem[EPRS(2020)]{ethicsofAI}
EPRS.
\newblock The ethics of artificial intelligence: Issues and initiatives.
\newblock \emph{European Parliament Research Service, Scientific Foresight Unit (STOA), Panel for the Future of Science and Technology}, 2020.

\bibitem[Qu et~al.(2023)Qu, Tang, and Ma]{qu2023adversarial}
Ao~Qu, Yihong Tang, and Wei Ma.
\newblock Adversarial attacks on deep reinforcement learning-based traffic signal control systems with colluding vehicles.
\newblock \emph{ACM Transactions on Intelligent Systems and Technology}, 14\penalty0 (6):\penalty0 1--22, 2023.

\bibitem[Xu et~al.(2024)Xu, Ling, and Fang]{xu2023learning}
Yixuan~Even Xu, Chun~Kai Ling, and Fei Fang.
\newblock Learning coalition structures with games.
\newblock In \emph{Proceedings of the AAAI Conference on Artificial Intelligence}, volume~38, pages 9944--9951, 2024.
\newblock \doi{10.1609/aaai.v38i9.28856}.
\newblock URL \url{https://ojs.aaai.org/index.php/AAAI/article/view/28856}.

\bibitem[Rosenthal(1973)]{rosenthal1973class}
Robert~W Rosenthal.
\newblock A class of games possessing pure-strategy nash equilibria.
\newblock \emph{International Journal of Game Theory}, 2:\penalty0 65--67, 1973.

\bibitem[Kearns et~al.(2001)Kearns, Littman, and Singh]{kearns2001graphical}
Michael Kearns, Michael~L Littman, and Satinder Singh.
\newblock Graphical models for game theory.
\newblock In \emph{Proceedings of the 17th Conference in Uncertainty in Artificial, Intelligence, 2001}, pages 253--260, 2001.

\bibitem[Fudenberg and Levine(1998)]{LearningGameBook}
Drew Fudenberg and David~K. Levine.
\newblock \emph{{The Theory of Learning in Games}}, volume~1 of \emph{MIT Press Books}.
\newblock The MIT Press, December 1998.

\bibitem[Waugh et~al.(2011)Waugh, Ziebart, and Bagnell]{WZB11}
Kevin Waugh, Brian~D. Ziebart, and J.~Andrew Bagnell.
\newblock Computational rationalization: the inverse equilibrium problem.
\newblock In \emph{Proceedings of the 28th International Conference on International Conference on Machine Learning}, ICML'11, page 1169–1176, Madison, WI, USA, 2011. Omnipress.
\newblock ISBN 9781450306195.

\bibitem[Kuleshov and Schrijvers(2015)]{KS15}
Volodymyr Kuleshov and Okke Schrijvers.
\newblock Inverse game theory: Learning utilities in succinct games.
\newblock In Evangelos Markakis and Guido Sch{\"a}fer, editors, \emph{Web and Internet Economics}, pages 413--427, Berlin, Heidelberg, 2015. Springer Berlin Heidelberg.

\bibitem[Ling et~al.(2018)Ling, Fang, and Kolter]{LFK18}
Chun~Kai Ling, Fei Fang, and J.~Zico Kolter.
\newblock What game are we playing? end-to-end learning in normal and extensive form games.
\newblock In \emph{Proceedings of the Twenty-Seventh International Joint Conference on Artificial Intelligence, {IJCAI-18}}, pages 396--402. International Joint Conferences on Artificial Intelligence Organization, 7 2018.
\newblock \doi{10.24963/ijcai.2018/55}.
\newblock URL \url{https://doi.org/10.24963/ijcai.2018/55}.

\bibitem[Letchford et~al.(2009)Letchford, Conitzer, and Munagala]{LCM09}
Joshua Letchford, Vincent Conitzer, and Kamesh Munagala.
\newblock Learning and approximating the optimal strategy to commit to.
\newblock In Marios Mavronicolas and Vicky~G. Papadopoulou, editors, \emph{Algorithmic Game Theory}, pages 250--262, Berlin, Heidelberg, 2009. Springer Berlin Heidelberg.

\bibitem[Balcan et~al.(2015)Balcan, Blum, Haghtalab, and Procaccia]{balcancommitment15}
Maria-Florina Balcan, Avrim Blum, Nika Haghtalab, and Ariel~D. Procaccia.
\newblock Commitment without regrets: Online learning in stackelberg security games.
\newblock In \emph{Proceedings of the Sixteenth ACM Conference on Economics and Computation}, EC '15, page 61–78, New York, NY, USA, 2015. Association for Computing Machinery.
\newblock ISBN 9781450334105.

\bibitem[Haghtalab et~al.(2016)Haghtalab, Fang, Nguyen, Sinha, Procaccia, and Tambe]{HaghtalabSecurityGame16}
Nika Haghtalab, Fei Fang, Thanh~H. Nguyen, Arunesh Sinha, Ariel~D. Procaccia, and Milind Tambe.
\newblock Three strategies to success: learning adversary models in security games.
\newblock In \emph{Proceedings of the Twenty-Fifth International Joint Conference on Artificial Intelligence}, IJCAI'16, page 308–314. AAAI Press, 2016.
\newblock ISBN 9781577357704.

\bibitem[Wu et~al.(2022)Wu, Shen, Fang, and Xu]{wu2022inverse}
Jibang Wu, Weiran Shen, Fei Fang, and Haifeng Xu.
\newblock Inverse game theory for stackelberg games: the blessing of bounded rationality.
\newblock \emph{Advances in Neural Information Processing Systems}, 35:\penalty0 32186--32198, 2022.

\bibitem[Bonjour et~al.(2022)Bonjour, Aggarwal, and Bhargava]{pmlr-v180-bonjour22a}
Trevor Bonjour, Vaneet Aggarwal, and Bharat Bhargava.
\newblock Information theoretic approach to detect collusion in multi-agent games.
\newblock In James Cussens and Kun Zhang, editors, \emph{Proceedings of the Thirty-Eighth Conference on Uncertainty in Artificial Intelligence}, volume 180 of \emph{Proceedings of Machine Learning Research}, pages 223--232. PMLR, 01--05 Aug 2022.

\bibitem[Mazrooei et~al.(2013)Mazrooei, Archibald, and Bowling]{Mazrooei_Archibald_Bowling_2013}
Parisa Mazrooei, Christopher Archibald, and Michael Bowling.
\newblock Automating collusion detection in sequential games.
\newblock \emph{Proceedings of the AAAI Conference on Artificial Intelligence}, 27\penalty0 (1):\penalty0 675--682, Jun. 2013.
\newblock \doi{10.1609/aaai.v27i1.8674}.
\newblock URL \url{https://ojs.aaai.org/index.php/AAAI/article/view/8674}.

\bibitem[Beigman and Vohra(2006)]{Beigman06}
Eyal Beigman and Rakesh Vohra.
\newblock Learning from revealed preference.
\newblock In \emph{Proceedings of the 7th ACM Conference on Electronic Commerce}, EC '06, page 36–42, New York, NY, USA, 2006. Association for Computing Machinery.
\newblock ISBN 1595932364.

\bibitem[Zadimoghaddam and Roth(2012)]{Zadimoghaddam12}
Morteza Zadimoghaddam and Aaron Roth.
\newblock Efficiently learning from revealed preference.
\newblock In Paul~W. Goldberg, editor, \emph{Internet and Network Economics}, pages 114--127, Berlin, Heidelberg, 2012. Springer Berlin Heidelberg.

\bibitem[Balcan et~al.(2014)Balcan, Daniely, Mehta, Urner, and Vazirani]{Balcan14}
Maria-Florina Balcan, Amit Daniely, Ruta Mehta, Ruth Urner, and Vijay~V. Vazirani.
\newblock Learning economic parameters from revealed preferences.
\newblock In Tie-Yan Liu, Qi~Qi, and Yinyu Ye, editors, \emph{Web and Internet Economics}, pages 338--353, Cham, 2014. Springer International Publishing.

\bibitem[Blum et~al.(2014)Blum, Haghtalab, and Procaccia]{BlumLearnOptCommit14}
Avrim Blum, Nika Haghtalab, and Ariel~D Procaccia.
\newblock Learning optimal commitment to overcome insecurity.
\newblock In Z.~Ghahramani, M.~Welling, C.~Cortes, N.~Lawrence, and K.Q. Weinberger, editors, \emph{Advances in Neural Information Processing Systems}, volume~27. Curran Associates, Inc., 2014.
\newblock URL \url{https://proceedings.neurips.cc/paper_files/paper/2014/file/cc1aa436277138f61cda703991069eaf-Paper.pdf}.

\bibitem[Amin et~al.(2015)Amin, Cummings, Dworkin, Kearns, and Roth]{Amin15}
Kareem Amin, Rachel Cummings, Lili Dworkin, Michael Kearns, and Aaron Roth.
\newblock Online learning and profit maximization from revealed preferences.
\newblock In \emph{Proceedings of the Twenty-Ninth AAAI Conference on Artificial Intelligence}, AAAI'15, page 770–776. AAAI Press, 2015.
\newblock ISBN 0262511290.

\bibitem[Roth et~al.(2016)Roth, Ullman, and Wu]{RothWatchandLearn16}
Aaron Roth, Jonathan Ullman, and Zhiwei~Steven Wu.
\newblock Watch and learn: optimizing from revealed preferences feedback.
\newblock In \emph{Proceedings of the Forty-Eighth Annual ACM Symposium on Theory of Computing}, STOC '16, page 949–962, New York, NY, USA, 2016. Association for Computing Machinery.
\newblock ISBN 9781450341325.

\bibitem[Cesa-Bianchi and Lugosi(2006)]{Cesa-Bianchi_Lugosi_2006}
Nicolo Cesa-Bianchi and Gabor Lugosi.
\newblock \emph{Prediction, Learning, and Games}.
\newblock Cambridge University Press, 2006.

\bibitem[Daskalakis et~al.(2021)Daskalakis, Fishelson, and Golowich]{daskalakis2021near}
Constantinos Daskalakis, Maxwell Fishelson, and Noah Golowich.
\newblock Near-optimal no-regret learning in general games.
\newblock \emph{Advances in Neural Information Processing Systems}, 34:\penalty0 27604--27616, 2021.

\bibitem[De~Bruijn(1981)]{de1981asymptotic}
Nicolaas~Govert De~Bruijn.
\newblock \emph{Asymptotic methods in analysis}, volume~4.
\newblock Courier Corporation, 1981.

\bibitem[Paes~Leme et~al.(2016)Paes~Leme, Pal, and Vassilvitskii]{paes2016field}
Renato Paes~Leme, Martin Pal, and Sergei Vassilvitskii.
\newblock A field guide to personalized reserve prices.
\newblock In \emph{Proceedings of the 25th international conference on world wide web}, pages 1093--1102, 2016.

\end{thebibliography}

\appendix
\clearpage

\section{Multiple-bit CSL with Congestion Games}
\label{sec:congestion_games}

To state our result for congestion games, we first recall the definition of a congestion game.

\begin{definition}[\cite{rosenthal1973class}]
	\label{def:congestion_game}
	A \textbf{congestion game} $\G$ is a tuple $(N, R, (A_i)_{i\in N}, (c_i)_{r\in R})$ where $N$ is the set of agents, $R$ is the set of resources, $A_i\subseteq 2^R \setminus \{\varnothing\}$ is the strategy space of agent $i$, and $c_r:\mathbb N \to \mathbb R$ is the cost function of resource $r$. For a strategy profile $\Sigma = (\sigma_i)_{i\in N}$ where $\sigma_i\in A_i$, the utility of agent $i$ is given by $u_i(\Sigma) = -\sum_{r\in \sigma_i} c_r(|\{j\in N\mid r\in \sigma_j\}|)$.
\end{definition}

Clearly, a directed prisoner's dilemma is not a congestion game, so we cannot directly apply \cref{alg:normal_form} to solve Multiple-bit CSL with congestion games. However, as we will show in this subsection, we can design another type of game-strategy pair, \textit{directed Braess's paradox}, which is a congestion game, and functions similarly to directed prisoner's dilemmas in the context of Multiple-bit CSL. We will then show that using directed Braess's paradoxes, we can solve Multiple-bit CSL with congestion games in $\log_2 n + 2$ rounds with an algorithm similar to \cref{alg:normal_form}.

\begin{definition}
    \label{def:directed_braess_paradox}
	For $i\in N, j\in N$, a \textbf{directed Braess's paradox} $\B(i,j)$ is a game-strategy pair $(\G, \Sigma)$. Here, $\G$ is a congestion game $(N, R, (A_i)_{i\in N}, (c_i)_{r\in R})$ where $R = \{r_1,r_2,r_3\}$. The strategy spaces of agent $i$ and $j$ are $\{\{r_1\}\}$ and $\{\{r_2\},\{r_1,r_3\}\}$ respectively, and the strategy spaces for all other agents are $\{\varnothing\}$. The cost functions are $c_{r_1}(x) = x$, $c_{r_2}(x) = 2.5$ and $c_{r_3}(x) = 0$ for all $x\in \mathbb N$. The strategy profile $\Sigma$ is such that $\sigma_i = \{r_1\}$, $\sigma_j = \{r_1,r_3\}$ and $\sigma_k = \varnothing$ for all $k\in N\setminus \{i,j\}$.
\end{definition}

We illustrate the definition of directed Braess's paradox in \cref{fig:directed_braess_paradox}.

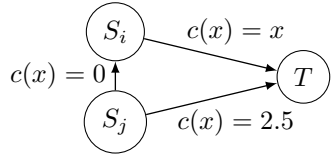
\begin{figure}[htbp]
    \centering
    \begin{tikzpicture}[scale=1]
        \node[circle, draw, minimum size=0.7cm] (i) at (0,0.6) {$S_i$};
        \node[circle, draw, minimum size=0.7cm] (j) at (0,-0.6) {$S_j$};
        \node[circle, draw, minimum size=0.7cm] (T) at (2.5,0) {$T$};
    
        \node at (1.6,-0.6) {$c(x) = 2.5$};
        \node at (1.6,0.6) {$c(x) = x$};

        \draw[->, >=latex, line width=0.5pt] (j) -- node[left] {$c(x) = 0$} (i);
        \draw[->, >=latex, line width=0.5pt] (i) -- (T);
        \draw[->, >=latex, line width=0.5pt] (j) -- (T);
    \end{tikzpicture}
    
    \caption{Illustration of the directed Braess's paradox $\B(i,j)$. Agent $i$ needs to choose a path from $S_i$ to $T$ and agent $j$ needs to choose a path from $S_j$ to $T$. The cost functions are indicated on the edges.}
	\label{fig:directed_braess_paradox}
\end{figure}

\begin{lemma}
	\label{lem:directed_braess_paradox}
	Let $i,j$ be two agents in $N$ and $\mathbf O = \O(\B(i,j))$. Then, $O_j = \mathrm{True}$ if and only if agent $i$ and agent $j$ are in the same coalition under $\S^*$. Moreover, $O_k = \mathrm{False}$ for $k\in N\setminus\{j\}$.
\end{lemma}

\begin{proof}
	If agent $i$ and agent $j$ are in the same coalition under $\S^*$, then choosing $\{r_2\}$ is a dominant strategy for agent $j$'s coalition. Otherwise, choosing $\{r_1,r_3\}$ is a dominant strategy. Therefore, $O_j = \mathrm{True}$ if and only if agent $i$ and agent $j$ are in the same coalition under $\S^*$. Moreover, since all agents other than $j$ only have one choice (either $\{r_1\}$ or $\{\varnothing\}$), $O_k = \mathrm{False}$ for $k\in N\setminus\{j\}$.
\end{proof}

\cref{lem:directed_braess_paradox} shows that directed Braess's paradoxes function just like directed prisoner's dilemmas in Multiple-bit CSL. With an argument similar to the proof from \cref{lem:directed_prisoner_dilemma} to \cref{lem:product_of_directed_prisoner_dilemma}, we can also show the following corollary about the product of multiple directed Braess's paradoxes.

\begin{corollary}
	\label{cor:product_of_directed_braess_paradox}
	Let $\{(i_x,j_x)\mid x \in \{1,2,\dots,k\}\}$ be $k$ pairs of agents and $\mathbf O = \O(\prod_{x=1}^{k} \B(i_x,j_x))$. Then, $O_j = \mathrm{True}$ if and only if there exists $x\in\{1,2,\dots,k\}\textbf{ such that }j_x = j$ and $i_x\in[j]_{\S^*}$.
\end{corollary}

To use the product of multiple directed Braess's paradoxes in an algorithm to solve Multiple-bit CSL with congestion games, we also need to show that it is still a congestion game. We will show that congestion games are closed under the product operation in the following lemma.

\begin{lemma}
	\label{lem:closedness_of_congestion_games}
	Let $(\G_1,\Sigma_1),(\G_2, \Sigma_2)$ be two game-strategy pairs and let $(\G_p,\Sigma_p) = (\G_1,\Sigma_1)\times (\G_2, \Sigma_2)$. If $\G_1$ and $\G_2$ are both congestion games on $N$, then $\G_p$ is also a congestion game on $N$.
\end{lemma}

\begin{proof}
	Let $\G_1 = (N, R_1, (A_{1,i})_{i\in N}, (c_{1,i})_{r\in R_1})$ and $\G_2 = (N, R_2, (A_{2,i})_{i\in N}, (c_{2,i})_{r\in R_2})$. We will write $\G_p$ as a congestion game. According to \cref{def:product_of_normal_form_games}, in $\G_p$, the strategy space of agent $i$ is $A_i = A_{1,i}\times A_{2,i}$. Moreover, given a strategy profile $\Sigma = (\sigma_{1,i},\sigma_{2,i})_{i\in N}$, the utility of agent $i$ is given by $u_i(\Sigma) = -\sum_{r\in \sigma_{1,i}} c_{1,r}(|\{j\in N\mid r\in \sigma_{1,j}\}|)-\sum_{r\in \sigma_{2,i}} c_{2,r}(|\{j\in N\mid r\in \sigma_{2,j}\}|)$. This shows that $\G_p = (N, (R_1,R_2), (A_{i})_{i\in N}, ((c_{1,i})_{r\in R_1}, \\ (c_{2,i})_{r\in R_2}))$ is a congestion game on $N$.
\end{proof}

With \cref{cor:product_of_directed_braess_paradox,lem:closedness_of_congestion_games}, we solve Multiple-bit CSL with congestion games just like we did with normal-form games in \cref{alg:normal_form}. 
The only difference is that we query directed Braess's paradoxes instead of directed prisoner's dilemmas. We present the algorithm below. 

\IncMargin{1.0em}
\begin{algorithm}[ht]

	\Indmm\Indmm
	\KwIn{The number of agents $n$ and the multiple-bit observation oracle $\O$.}
	\KwOut{A coalition structure $\S$ of the agents.}
	\Indpp\Indpp
	
	Query $\prod_{i\in N,j\in N, i<j}\B(i,j)$ and observe $\mathbf O$\;
	Let $T_j \gets \{1,2,\dots,j-1\}\textbf{ if } O_j=\mathrm{True}\textbf{ else }\varnothing$ $\textbf{ for each }j\in N$\;
	\While{$\exists j \in N\textbf{ such that }|T_j|\geq 2$}{
		Let $L_j \gets \{\text{the smallest $\lfloor\frac{|T_j|}2\rfloor$ elements in $T_j$}\}$ $\textbf{ for each }j\in N$\;
		Let $R_j \gets T_j \setminus L_j$\textbf{ for each }$j\in N$\;
		Query $\prod_{j\in N,i\in L_j}\B(i,j)$ and observe $\mathbf O$\;
		Let $T_j \gets L_j\textbf{ if } O_j=\mathrm{True}\textbf{ else } R_j $ $\textbf{ for each }j\in N$\;
	}
	Let $\S\gets\{\{1\},\{2\},\dots,\{n\}\}$\;
	\For{$j\in N\textbf{ such that }|T_j|\ne \varnothing$}{
		Let $i\gets\text{the only element in $T_j$}$\;
		Merge $[i]_\S$ and $[j]_\S$ in $\S$\;
	}
	\Return{$\S$}\;
	\caption{Simultaneous Binary Search with Congestion Games}
	\label{alg:congestion}
\end{algorithm}
\DecMargin{1.0em}

\begin{restatable}{theorem}{congestion}
	\label{thm:congestion}
	Algorithm \ref{alg:congestion} solves Multiple-bit CSL with congestion games in $\log_2 n + 2$ rounds.
\end{restatable}

\begin{proof}
	The proof is almost identical to the proof of \cref{thm:normal_form}. We first show the correctness of the algorithm. To do this, we will show three claims: (i) after Lines 1 to 7,  for each $j\in N$, if $j$ is the agent with the smallest index in $[j]_{S^*}$, $T_j=\varnothing$, otherwise, $T_j$ contains only the smallest index of the agents in $[j]_{S^*}$; (ii) after Lines 8 to 11, $\S = \S^*$; (iii) the games queried in the algorithm are congestion games. (iii) is implied by \cref{lem:closedness_of_congestion_games}, so we only need to show (i) and (ii).
	
	For (i), if $j$ is the agent with the smallest index in $[j]_{S^*}$, then by \cref{cor:product_of_directed_braess_paradox}, $O_j = \mathrm{False}$, and $T_j = \varnothing$ on Line 2. Throughout the algorithm, $T_j$ remains $\varnothing$. Otherwise, $T_j$ contains the smallest index of the agents in $[j]_{S^*}$ after Line 2. In the while loop, $T_j$ is updated to $L_j$ if $O_j = \mathrm{True}$, and to $R_j$ otherwise. Since $L_j$ contains the smallest $\lfloor|T_j|/2\rfloor$ elements in $T_j$, by \cref{cor:product_of_directed_braess_paradox}, $O_j = \mathrm{True}$ if and only if the smallest index of the agents in $[j]_{S^*}$ is in $L_j$. Therefore, after one iteration of the loop, $T_j$ still contains only the smallest index of the agents in $[j]_{S^*}$, while the size of $T_j$ is halved. The loop terminates after $|T_j|=1$, and $T_j$ contains only the smallest index of the agents in $[j]_{S^*}$.

	For (ii), since every agent $j$ is either the agent with the smallest index in $[j]_{S^*}$ or merged with that agent, each coalition is merged in $\S$ after Lines 8 to 11. Thus $\S$ becomes the same as $\S^*$.

	Next, we show the complexity of the algorithm. The while loop in Lines 3 to 7 runs at most $\lceil\log_2 n\rceil$ times, and each iteration requires $1$ query. Together with the query on Line 1, the total number of queries is at most $\lceil\log_2 n\rceil + 1 \leq \log_2 n + 2$.
\end{proof}

\section{Multiple-bit CSL with Graphical Games}
\label{sec:graphical_games}

We then proceed to the case of Multiple-bit CSL with graphical games. Recall that in a graphical game \cite{kearns2001graphical}, each agent $i$ is associated with a vertex in an undirected graph $G$. We use $N_i \subseteq N$ to denote the set of neighbors of agent $i$ in $G$, including agent $i$ itself.

\begin{definition}
	\label{def:graphical_game}
	A \textbf{graphical game} $\G$ is a tuple $(N, G, (A_i)_{i\in N}, (M_i)_{i\in N})$ where $N$ is the set of agents, $G$ is an undirected graph, $A_i$ is the strategy space of agent $i$, and $M_i: \prod_{j\in N_i} A_j \to \mathbb R$ is the \textbf{local game matrix} of agent $i$. For a strategy profile $\Sigma = (\sigma_i)_{i\in N}$ where $\sigma_i\in A_i$, the utility of agent $i$ is given by $u_i(\Sigma) = M_i(\Sigma_{N_i})$, where $\Sigma_{N_i} = (\sigma_j)_{j\in N_i}$ consists only of the strategies of agents in $N_i$. Here, $N_i$ denotes the set of neighbors of agent $i$ in undirected graph $G$.
\end{definition}

If we let the undirected graph $G$ be a complete graph, then we recover the normal-form game setting. However, the reason why we consider Multiple-bit CSL with graphical games is that graphical games are compact representations of normal-form games. When the degree of each vertex in the graph is small, the graphical game representation is far succincter than the normal-form game representation. Therefore, we consider in this section the case where the algorithm can only design degree-$d$ graphical games, where the degree of each vertex in the graph is at most $d$. In this case, we will need at least $\lceil\frac{n-1}{d}\rceil$ rounds to solve the Multiple-bit CSL problem. To prove this claim formally, we first show with the following lemma that in some cases, the observations reveal no information about whether a certain pair of agents are in the same coalition under $\S^*$. 

	\begin{lemma}
		\label{lem:indistinguishable_graphical}
		Let $i,j$ be two agents in $N$ and let $\S_0 = \{\{1\},\{2\},\dots,\{n\}\}, \S_1 = \S_0\cup\{\{i,j\}\}\setminus\{\{i\},\{j\}\}$ be two coalition structures. If we query the oracle for a graphical game $\G$ where edge $(i,j)$ is not in the undirected graph $G$, then the observation reveals no information about whether $\S^* = \S_0$ or $\S^* = \S_1$. 
	\end{lemma}
	
	\begin{proof}
		Let $\G = (N, G, (A_i)_{i\in N}, (M_i)_{i\in N})$ and let $\Sigma = (\sigma_i)_{i\in N}$ be the specified strategy profile in $\G$. Consider agent $i$'s deviation decision $\O_i(\G, \Sigma)$. Let $\BR_{i,\S_0}(\{i\},\Sigma)$ and $\BR_{j,\S_0}(\{j\},\Sigma)$ be the set of best responses of coalitions $\{i\}$ and $\{j\}$ to $\Sigma_{-i}$ and $\Sigma_{-j}$ if $\S^* = \S_0$ respectively. Since the local game matrix $M_i$ only depends on $\Sigma_{N_i}\not \ni j$, no matter what strategy agent $j$ plays, the best response of $\{i\}$ to $\Sigma_{-i}$ is the same. Therefore, the set of best responses of coalition $\{i,j\}$ if $\S^* = \S_1$ is $\BR_{i,\S_0}(\{i\},\Sigma) \times \BR_{j,\S_0}(\{j\},\Sigma)$. This shows that whether $\S^* = \S_0$ or $\S^* = \S_1$ does not affect the deviation decision of agent $i$. The same argument applies to agent $j$.
	\end{proof}

\begin{theorem}
	\label{thm:lower_bound_graphical}
	Any algorithm that solves the Multiple-bit CSL problem with degree-$d$ graphical games requires at least $\lceil\frac{n-1}{d}\rceil$ rounds of interactions with the agents in the worst case.
\end{theorem}

\begin{proof}
	Suppose that the algorithm has queried the observation oracle $k$ times with game-strategy pairs $\{(\G_1,\Sigma_1),(\G_2,\Sigma_2),\dots,(\G_k,\Sigma_k)\}$, where $\G_i = (N, G_i, (A_i)_{i\in N}, (M_i)_{i\in N})$ is a degree-$d$ graphical game. 
	According to \cref{lem:indistinguishable_graphical}, if the observations always behave as if $\S^* = \{\{1\},\{2\},\dots,\{n\}\}$, then the algorithm cannot finalize the answer until for each pair of agents $i,j$, the algorithm has queried the oracle for a graphical game where edge $(i,j)$ is in the undirected graph. Since the degree of each vertex in the graph is at most $d$, the algorithm needs to query the oracle for at least $\lceil\frac{n-1}{d}\rceil$ different graphical games to include all edges adjacent to vertex $1$. The theorem follows.
\end{proof}

Then, we present an algorithm that solves the Multiple-bit CSL problem with degree-$d$ graphical games in $\frac{2n}{d}+2\log_2d + 1$ rounds of interactions with the agents in the worst case. Note that $d\leq n$, given the lower bounds in \cref{thm:lower_bound,thm:lower_bound_graphical}, the algorithm is optimal up to a constant factor.

The algorithm uses the product of multiple directed prisoner's dilemmas to query the observation oracle. We will first show that such a product can be represented as a degree-$d$ graphical game.

\begin{lemma}
	\label{lem:degree_of_product_of_directed_prisoner_dilemma}
	Let $\{(i_x,j_x)\mid x \in \{1,2,\dots,k\}\}$ be $k$ pairs of agents and $(\G_p,\Sigma_p) = \prod_{x=1}^{k} \P(i_x,j_x)$. Then, $\G_p$ can be written as a degree-$d$ graphical game, where $d = \max\{\sum_{x=1}^k(\ind{i = i_x} + \ind{i = j_x})\mid i\in N\}$, i.e., the maximum number of occurrences of any agent in the pairs.
\end{lemma}

\begin{proof}
	For two agents $i\in N, j\in N$, if $(i,j)$ does not occur in any of the pairs, then, according to \cref{def:directed_prisoner_dilemma}, the utility of agent $i$ does not depend on agent $j$'s strategy in any of the directed prisoner's dilemmas $\P(i_x,j_x) \mid x\in \{1,2,\dots, k\}$. Since for each agent, the utility in the product game is the sum of the utilities in the factor games, the utility of agent $i$ in the product game also does not depend on agent $j$'s strategy. This shows that $\G_p$ can be represented as a graphical game with the graph's edge set being $\{(i_x,j_x)\mid x \in \{1,2,\dots,k\}\}$. The lemma then follows.
\end{proof}

We then present the algorithm as \cref{alg:graphical}.

\IncMargin{1.0em}
\begin{algorithm}[ht]

	\Indmm\Indmm
	\KwIn{The number of agents $n$, the degree limit $d$ and the multiple-bit observation oracle $\O$.}
	\KwOut{A coalition structure $\S$ of the agents.}
	\Indpp\Indpp
	
	Let $size \gets \lfloor\frac{d}{2}\rfloor$ and $cnt \gets \lceil \frac{n}{size} \rceil$\;
	Let $belong_j \gets \lfloor\frac{j - 1}{size}\rfloor\textbf{ for each }j \in N$\;
	\For{$\delta \in \{0,1,\dots,cnt-1\}$}{
		Query $\prod_{i\in N,j\in N, i<j,belong_j-belong_i=\delta}\P(i,j)$ and observe $\mathbf O^{(\delta)}$\;
	}
	Let $S_j \gets \{\delta \mid O^{(\delta)}_j = \mathrm{True}\}\textbf{ for each }j \in N$\;
	Let $\delta_j \gets -1\textbf{ if } S_j=\varnothing\textbf{ else }\min(S_j)$ $\textbf{ for each }j\in N$\;
	Let $T_j \gets \{i \mid i < j, belong_j - belong_i = \delta_j\}\textbf{ if } \delta_j \ne -1\textbf{ else }\varnothing\textbf{ for each }j\in N$\;
	\While{$\exists j \in N\textbf{ such that }\delta_j = 0, |T_j|\geq 2$}{
		Let $L_j \gets \{\text{the smallest $\lfloor\frac{|T_j|}2\rfloor$ elements in $T_j$}\}$\textbf{ for each }$j\in N,\delta_j = 0$\;
		Let $R_j \gets T_j \setminus L_j$\textbf{ for each }$j\in N,\delta_j = 0$\;
		Query $\prod_{j\in N,\delta_j = 0,i\in R_j}\P(i,j)$ and observe $\mathbf O$\;
		Let $T_j \gets R_j\textbf{ if } O_j=\mathrm{True}\textbf{ else } L_j $ $\textbf{ for each }j\in N,\delta_j = 0$\;
	}
	\While{$\exists j \in N\textbf{ such that }\delta_j \geq 1, |T_j|\geq 2$}{
		Let $L_j \gets \{\text{the smallest $\lfloor\frac{|T_j|}2\rfloor$ elements in $T_j$}\}$\textbf{ for each }$j\in N,\delta_j \geq 1$\;
		Let $R_j \gets T_j \setminus L_j$\textbf{ for each }$j\in N,\delta_j \geq 1$\;
		Query $\prod_{j\in N,\delta_j \geq 1,i\in R_j}\P(i,j)$ and observe $\mathbf O$\;
		Let $T_j \gets R_j\textbf{ if } O_j=\mathrm{True}\textbf{ else } L_j $ $\textbf{ for each }j\in N,\delta_j \geq 1$\;
	}
	Let $\S\gets\{\{1\},\{2\},\dots,\{n\}\}$\;
	\For{$j\in N\textbf{ such that }|T_j|\ne \varnothing$}{
		Let $i\gets\text{the only element in $T_j$}$\;
		Merge $[i]_\S$ and $[j]_\S$ in $\S$\;
	}
	\Return{$\S$}\;
	\caption{Block Decomposition and Simultaneous Binary Search with Graphical Games}
	\label{alg:graphical}
\end{algorithm}
\DecMargin{1.0em}

Compared with the normal-form game setting, with degree-$d$ graphical games, the algorithm can only query products of multiple directed prisoner's dilemmas subject to the degree constraint $d$ as described in \cref{lem:degree_of_product_of_directed_prisoner_dilemma}. To comply with this constraint, \cref{alg:graphical} introduces an important idea, \textit{block decomposition}, i.e., to partition the agents into blocks of size $\lfloor\frac{d}{2}\rfloor$ and choose pairs to query according to the decomposition (Lines 1 to 2). As long as for each query, each agent is involved in directed prisoner's dilemmas with at most two blocks of agents, the constraint is not violated. 

The general idea of \cref{alg:graphical} is to find for each agent $j$, the ``predecessor'' of $j$, i.e., the agent with the largest index that is smaller than $j$ in $[j]_{\S^*}$. To do this, the algorithm takes a two-step approach. First, the algorithm tries to find which block the predecessor of each agent belongs to (Lines 3 to 6). We illustrate this process with an example in \cref{fig:decomposition}. The algorithm enumerates \textit{block index gap} $\delta$ and constructs one query for each $\delta$. The query for a given $\delta$ includes $\P(i,j)$ of all pairs of agents $i, j$ such that the difference between the block indices they belong to, i.e., $belong_j - belong_i$, is $\delta$ (Lines 3 to 4). In this way, an agent $j$ will be involved in directed prisoner's dilemmas with agents in at most two blocks, block $belong_j + \delta$ and block $belong_j - \delta$ in one query. Moreover, after the enumeration, since the algorithm observes for each agent $i$ and each block before $belong_i$, whether there is an agent in the same coalition of $i$ in that block, it gets to determine the block index of each agent's predecessor or conclude that it has no predecessors (Lines 5 to 6). 

Then, the algorithm carries out a simultaneous binary search to find the predecessor of each agent (Lines 7 to 17). The algorithm first deals with the agents whose predecessors are in the same block as themselves (Lines 7 to 12). In this process, each agent $j$ will be involved in games only within its own block in the queries. Next, the algorithm deals with the agents whose predecessors are in a different block (Lines 13 to 17). In this process, each agent $j$ will be involved in games with its predecessor's block, and agents outside block $belong_j$ whose predecessors are in block $belong_j$. The number of such agents is also at most twice the block size.

Finally, the algorithm merges each agent with its predecessor (Lines 18 to 21).

\begin{figure*}[htbp]
    \centering
    
    \begin{subfigure}[b]{0.32\textwidth}
        \centering
        \begin{tikzpicture}[scale=1]
            \node[circle, draw] (1) at (0,1.2) {1};
            \node[circle, draw] (3) at (1.5,1.7) {3};
            \node[circle, draw] (5) at (3,1.2) {5};
            \node[circle, draw] (2) at (0,0) {2};
            \node[circle, draw] (4) at (1.5,-0.5) {4};
            \node[circle, draw] (6) at (3,0) {6};

            \node at (0,1.8) {Block $0$};
            \node at (1.5,2.3) {Block $1$};
            \node at (3,1.8) {Block $2$};

            \draw[->, >=latex, line width=0.5pt] (2) -- (1);
            \draw[->, >=latex, line width=0.5pt] (4) -- (3);
            \draw[->, >=latex, line width=0.5pt] (6) -- (5);
        \end{tikzpicture}
        \caption{When $\delta = 0$}
		\label{subfigure:decomposition_a}
    \end{subfigure}
    \begin{subfigure}[b]{0.32\textwidth}
        \centering
        \begin{tikzpicture}[scale=1]
            \node[circle, draw] (1) at (0,1.2) {1};
            \node[circle, draw] (3) at (1.5,1.7) {3};
            \node[circle, draw] (5) at (3,1.2) {5};
            \node[circle, draw] (2) at (0,0) {2};
            \node[circle, draw] (4) at (1.5,-0.5) {4};
            \node[circle, draw] (6) at (3,0) {6};

            \node at (0,1.8) {Block $0$};
            \node at (1.5,2.3) {Block $1$};
            \node at (3,1.8) {Block $2$};

            \draw[->, >=latex, line width=0.5pt] (3) -- (1);
            \draw[->, >=latex, line width=0.5pt] (3) -- (2);
            \draw[->, >=latex, line width=0.5pt] (4) -- (1);
            \draw[->, >=latex, line width=0.5pt] (4) -- (2);
            \draw[->, >=latex, line width=0.5pt] (5) -- (3);
            \draw[->, >=latex, line width=0.5pt] (5) -- (4);
            \draw[->, >=latex, line width=0.5pt] (6) -- (3);
            \draw[->, >=latex, line width=0.5pt] (6) -- (4);
        \end{tikzpicture}
        \caption{When $\delta = 1$}
		\label{subfigure:decomposition_b}
    \end{subfigure}
    \begin{subfigure}[b]{0.32\textwidth}
        \centering
        \begin{tikzpicture}[scale=1]
            \node[circle, draw] (1) at (0,1.2) {1};
            \node[circle, draw] (3) at (1.5,1.7) {3};
            \node[circle, draw] (5) at (3,1.2) {5};
            \node[circle, draw] (2) at (0,0) {2};
            \node[circle, draw] (4) at (1.5,-0.5) {4};
            \node[circle, draw] (6) at (3,0) {6};

            \node at (0,1.8) {Block $0$};
            \node at (1.5,2.3) {Block $1$};
            \node at (3,1.8) {Block $2$};

            \draw[->, >=latex, line width=0.5pt] (5) -- (1);
            \draw[->, >=latex, line width=0.5pt] (5) -- (2);
            \draw[->, >=latex, line width=0.5pt] (6) -- (1);
            \draw[->, >=latex, line width=0.5pt] (6) -- (2);
        \end{tikzpicture}
        \caption{When $\delta = 2$}
		\label{subfigure:decomposition_c}
    \end{subfigure}
    
    \caption{Example execution of Lines 1 to 6 in \cref{alg:graphical} when $n = 6$ and $size = 2$. The vertices represent the agents and the edges represent the directed prisoner's dilemmas that the algorithm queries for each $\delta = \{0,1,2\}$. The algorithm first partitions the agents into $cnt = \lceil\frac{n}{size}\rceil = 3$ blocks, each containing at most $size = 2$ agents (Lines 1 to 2). Then, for each $\delta = \{0,1,2\}$, the algorithm queries the oracle for the directed prisoner's dilemmas that correspond to the edges shown in the figure (Lines 3 to 4). Using the observations, the algorithm then determines for each agent $j$, $\delta_j = belong_j - belong_i$, where $i$ is $j$'s predecessor in $[j]_{\S^*}$ (Lines 5 to 6). Note that in each query, any agent $j$ is involved in at most $2size\leq d$ games.}
	\label{fig:decomposition}
\end{figure*}
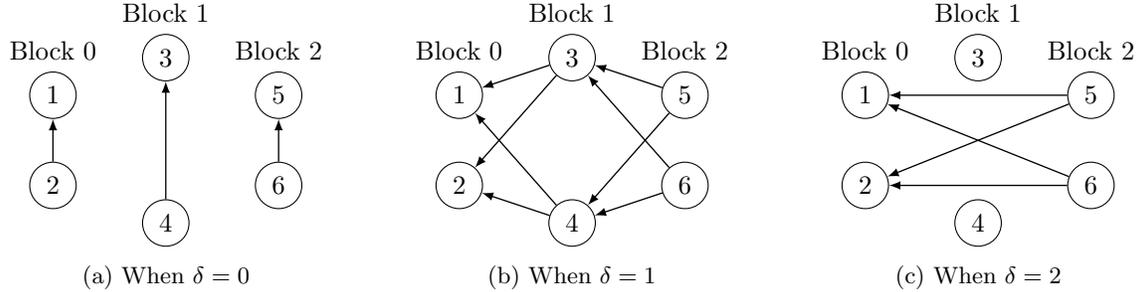

\begin{theorem}
	\label{thm:graphical}
	Let $d\geq 2$ be an even number, \cref{alg:graphical} solves the Multiple-bit CSL problem with degree-$d$ graphical games in $\frac{2n}{d}+2\log_2d + 1$ rounds of interactions with the agents in the worst case.
\end{theorem}

\begin{proof}
	We first show the correctness of the algorithm. To do this, we will show four claims: (i) after Lines 1 to 6, for each $j\in N$, if $j$ is the agent with the smallest index in $[j]_{S^*}$, $\delta_j = -1$; otherwise, let $i$ be the largest index in $[j]_{S^*}$ that is smaller than $j$, $\delta_j = belong_j - belong_i$; (ii) after Lines 7 to 17, for each $j\in N$, if $j$ is the agent with the smallest index in $[j]_{S^*}$, $T_j=\varnothing$; otherwise, $T_j$ contains only the largest index in $[j]_{S^*}$ that is smaller than $j$; (iii) after Lines 8 to 21, $\S = \S^*$; (iv) the games queried in Lines 4, 11, and 16 are degree-$d$ graphical games.

	For (i), if $j$ is the agent with the smallest index in $[j]_{S^*}$, then by \cref{lem:product_of_directed_prisoner_dilemma}, $O^{(\delta)}_j = \mathrm{False}$, for all $\delta \in \{0,1,\dots, cnt - 1\}$ on Line 4. As a result, $S_j = \varnothing$ on Line 5, and $\delta_j = -1$ on Line 6. Otherwise, let $i$ be the largest index in $[j]_{S^*}$ that is smaller than $j$. Then, by \cref{lem:product_of_directed_prisoner_dilemma}, $O^{(\delta)}_j = \mathrm{True}$ for $\delta = belong_j - belong_i$ on Line 4 since $\P(i,j)$ is included as a factor game. Moreover, $O^{(\delta)}_j = \mathrm{False}$ for all $\delta < belong_j - belong_i$ since $i$ is the largest index in $[j]_{S^*}$ that is smaller than $j$. Therefore, $belong_j - belong_i$ becomes the minimum of $S_j$ on Line 5, and thus $\delta_j = belong_j - belong_i$ on Line 6.
	
	For (ii), if $j$ is the agent with the smallest index in $[j]_{S^*}$, then $T_j = \varnothing$ on Line 7. Throughout the algorithm, $T_j$ remains $\varnothing$. Otherwise, $T_j$ contains the largest index, $i$, in $[j]_{S^*}$ that is smaller than $j$ after Line 7. Depending on whether $belong_j - belong_i = 0$ or not, either Lines 8 to 12 or Lines 13 to 17 will carry out a binary search to find this index $i$. Note that in the while loop, $T_j$ is updated to $R_j$ if $O_j = \mathrm{True}$, and to $L_j$ otherwise. Since $R_j$ contains the largest $\lceil|T_j|/2\rceil$ elements in $T_j$, by \cref{lem:product_of_directed_prisoner_dilemma}, $O_j = \mathrm{True}$ if and only if the largest index $i$ in $[j]_{S^*}$ that is smaller than $j$ is in $R_j$. Thus, the binary search will find this index $i$ and end up with $T_j = \{i\}$.

	For (iii), since every agent $j$ is either the agent with the smallest index in $[j]_{S^*}$ or merged with its predecessor (the agent with the largest index that is smaller than $j$) in $[j]_{S^*}$, each coalition is merged in $\S$ after Lines 18 to 21. Thus $\S$ becomes the same as $\S^*$.

	For (iv), consider an agent $j$. In the game queried on line 4, agent $j$ is involved in at most $2size \leq d$ factor games
    with agents $\{i\mid |belong_j - belong_i| = \delta\}$; in the game queried on line 11, agent $j$ is involved in at most $size \leq d$ factor games with agents $\{i\mid belong_i = belong_j\}$; in the game queried on line 16, agent $j$ is involved in at most $2size \leq d$ factor games with agents $\{i \mid i < j, belong_j - belong_i = \delta_j\}\cup\{i \mid i > j, \delta_i \geq 1, belong_i - belong_j = \delta_i\}$. Therefore, the games queried in Lines 4, 11, and 16 are degree-$d$ graphical games by \cref{lem:degree_of_product_of_directed_prisoner_dilemma}.

	Next, we show the complexity of the algorithm. The for loop in Lines 3 to 4 requires $cnt = \lceil\frac{n}{size}\rceil = \lceil\frac{2n}{d}\rceil$ queries, and the while loops in Lines 8 to 12 and Lines 13 to 17 require at most $\lceil\log_2 size\rceil = \lceil\log_2 d\rceil-1$ queries each. Therefore, the total number of queries is at most $\lceil\frac{2n}{d}\rceil + 2\lceil\log_2 d\rceil - 2 \leq \frac{2n}{d} + 2\log_2 d + 1$. This completes the proof.
\end{proof}

\section{Omitted Proofs}
\label{app:missing_parts}

\subsection{Proof of Theorem \ref{thm:auction_linear}}
\label{appsub:proof_of_thm_auction_linear}

\auctionlinear*

\begin{proof}
	We first show the correctness of the algorithm. Consider the if statement on Line 8, if there exists $j$ such that $O_j = \mathrm{True}$, then by \cref{lem:auction_gadget}, exactly one agent $j$ in $T$ has $O_j = \mathrm{True}$, and agents $i$ and $j$ are in the same coalition under $\S^*$. The algorithm then merges $[i]_\S$ and $[j]_\S$ in $\S$, and removes $j$ from $T$. Otherwise, by \cref{lem:auction_gadget}, no agent in $T$ is in $[i]_{\S^*}$. The algorithm then removes $i$ from $T$. Therefore, the inner while loop on Lines 5 to 12 finalizes agent $i$'s coalition in $\S$ correctly, and removes the whole coalition from $T$. The correctness of the algorithm then follows.
	
	For the complexity, the algorithm starts with $T = N$, and after each query on Line 7, it removes one agent from $T$. Since the algorithm terminates when $|T| = 1$, it uses $n - 1$ queries in total.
\end{proof}

\end{document}